\documentclass[11pt,twoside]{article}  

\usepackage[utf8]{inputenc}
\usepackage[margin = 2.5cm]{geometry}
\usepackage{amsmath,amssymb,amsfonts,amsthm}
\usepackage{appendix}
\usepackage[dvipsnames]{xcolor}
\usepackage{siunitx}
\usepackage{graphicx}
\usepackage[numbers]{natbib}
\RequirePackage[colorlinks,citecolor=blue,urlcolor=blue, allcolors=blue]{hyperref}
\usepackage{newpxtext,newpxmath} 

\usepackage{authblk}

\usepackage{enumitem}
\usepackage{subcaption}
\usepackage{xspace}
\usepackage{listings}

\usepackage{rotating}

\usepackage{booktabs}

\usepackage{graphicx}

\usepackage{array}
\usepackage{makecell}

\theoremstyle{plain}
\newtheorem{lem}{Lemma}
\newtheorem{thm}{Theorem}

\newtheorem{prop}{Proposition}

\usepackage[linesnumbered,ruled,vlined]{algorithm2e}

\theoremstyle{definition} 
\newtheorem{defn}{Definition}
\newtheorem{examplex}{Example}
\newenvironment{ex}
  {\pushQED{\qed}\examplex}
  {\popQED\endexamplex}

\newcommand{\Z}{\mathbb{Z}}
\newcommand{\R}{\mathbb{R}}

\renewcommand{\P}{\mathbb{P}}
\newcommand{\E}{\mathbb{E}}

\renewcommand{\L}{\mathbb{L}}
\renewcommand{\O}{\mathbb{O}}

\usepackage{bbm}
\usepackage{caption}

\usepackage{mathtools}
\DeclarePairedDelimiter{\ceil}{\lceil}{\rceil}

\title{Enabling Humanitarian Applications with \\Targeted Differential Privacy}

\author[1]{Nitin Kohli}
\author[1,2,*]{Joshua E. Blumenstock}

\affil[1]{Center for Effective Global Action, UC Berkeley}
\affil[2]{School of Information, UC Berkeley}

\affil[*]{Corresponding Author: jblumenstock@berkeley.edu}

\date{}

\begin{document}

\maketitle

\begin{abstract}
The proliferation of mobile phones in low- and middle-income countries has suddenly and dramatically increased the extent to which the world’s poorest and most vulnerable populations can be observed and tracked by governments and corporations. Millions of historically ``off the grid'' individuals are now passively generating digital data; these data, in turn, are being used to make life-altering decisions about those individuals --- including whether or not they receive government benefits, and whether they qualify for a consumer loan. 
This paper develops an approach to implementing algorithmic decisions based on personal data, while also providing formal privacy guarantees to data subjects. The approach adapts differential privacy to applications that require decisions about individuals, and gives decision makers granular control over the level of privacy guaranteed to data subjects. We show that stronger privacy guarantees typically come at some cost, and use data from two real-world applications --- an anti-poverty program in Togo and a consumer lending platform in Nigeria --- to illustrate those costs. Our empirical results quantify the tradeoff between privacy and predictive accuracy, and characterize how different privacy guarantees impact overall program effectiveness. More broadly, our results demonstrate a way for humanitarian programs to responsibly use personal data, and better equip program designers to make informed decisions about data privacy.
\end{abstract}


\section*{Introduction} 
\label{sec:intro}
Mobile phones are now close to ubiquitous in even the world's poorest nations. It is estimated that 81\% of women and 87\% of men in low- and middle-income countries (LMICs) own a mobile phone \cite{jeffrie2023mobile}, and roughly 150 million individuals used mobile internet  for the first time in 2023 \cite{shanahan2023state}. With continual advances and investments in digital public infrastructure \cite{eaves2019government, hong2023dpi} and network connectivity, these numbers are expected to rise \cite{gsma2024mobileeconomy}. Consequently, billions of historically ``off the grid'' individuals are now passively generating digital footprint data through the everyday use of their mobile devices. 

These data, in turn, are being used by governments and corporations to make life-altering decisions about individuals in a variety of humanitarian applications. For example, digital footprint data from mobile phone networks have been successfully used in social protection and anti-poverty initiatives to identify individuals with the greatest need for humanitarian support \cite{aiken2022machine}. Mobile phone metadata have likewise been used in consumer credit products across LMICs to determine loan eligibility for hundreds of millions of historically unbanked individuals \cite{robinson2023impact}. In each of these settings, people's digital footprint data are combined with machine learning algorithms to \textit{target} them, i.e., to determine whether they should be eligible for some benefit (e.g., cash transfers or loans). 

However, the metadata generated when people use their mobile phones is exceptionally sensitive. Prior work has shown that mobile phone metadata can reveal an individual's sexual orientation \cite{ovide2021, boorstein2021},  religious affiliation \cite{dube2022measuring}, political preferences \cite{thompson2019twelve}, and social network connections \cite{eagle2009inferring}. For this reason, the use and analysis of these data --- particularly from socioeconomically disadvantaged populations --- raises serious privacy concerns \cite{taylor2015name, mann2018left, taylor2020price, blumenstock2023big}. 

This paper develops and tests a novel approach to making algorithmic targeting decisions based on provably private data, which provides formal and robust privacy guarantees while simultaneously enabling accurate interventions in downstream applications.  We make three main contributions. Our first contribution is methodological and provides a framework for rigorously reasoning about privacy in targeting applications. We illustrate how a variant of differential privacy --- which we call \textit{targeted differential privacy} --- can be used to produce provably private datasets that still allow for accurate targeting. Our second contribution is technical and provides a novel algorithm to generate privatized datasets that satisfy targeted differential privacy. Finally, we provide robust evidence on the empirical tradeoff between privacy and predictive accuracy in two real-world settings: an anti-poverty program in Togo and a consumer lending platform in Nigeria. The analysis shows how higher levels of privacy can protect against two canonical threats (singling-out attacks and attribute inference attacks), but that such privacy protections can impact overall program effectiveness. 
Perhaps most notably, we find that large increases in privacy can often be obtained for relatively small sacrifices in targeting accuracy. Taken together, our results  illustrate how mathematical techniques can enable the responsible use of personal data in humanitarian programs, and better equip program designers to make informed decisions about the tradeoffs involved in using targeted differential privacy in practice.

\section*{A framework and algorithm for provably private targeting}
\label{sec:approach}
Our first contribution is a framework for rigorously reasoning about data privacy in targeting applications. Underlying this framework is an interaction between a data holder (such as a mobile network operator) and a downstream policymaker or program official who wishes to make use of the data (such as a humanitarian program manager or loan officer). The program official seeks to target a policy (or product, promotion, etc.) to individuals based on eligibility criteria that are difficult or costly to observe directly, such as individual levels of consumption, deprivation, credit-worthiness, or profitability. (This is often the case in humanitarian settings when, for instance, the government does not have accurate and up-to-date data on impacted populations \cite{lindert2020sourcebook}).  We consider the common case in which the policymaker wishes to use a decision rule to determine eligibility based on \textit{proxy} information, where the proxies that form the decision rule are learned from patterns observed in a subset of the population for whom eligibility data can be directly observed (i.e., a learning or training sample). In traditional anti-poverty programs, administrators use simple linear models to identify observable household characteristics (such as roofing material) that can proxy for more nuanced measures of welfare (such as food consumption) \cite{grosh1995proxy}; in more recent programs, sophisticated machine learning algorithms are applied to high-dimensional data from multiple sources to predict household wealth \citep{aiken2022machine}. Likewise, traditional lending decisions are based on formal financial histories; more recent ``alternative'' credit scores use data from mobile phones and social media \citep{bjorkegren2020behavior}.

We study the privacy implications of this interaction when the program official wishes to base decisions on proxy data $X$ that are owned by the data holder, where $X$ contains sensitive information. To protect the privacy of the data, the data holder provides the program official a privacy-preserving dataset $X_{priv}$ that is generated by some privacy-enhancing technology. The program official then uses $X_{priv}$ as the basis for a decision rule that determines program eligibility. We refer to this data sharing interaction and prediction process as the \textit{targeting setting}; it is depicted in Figure \ref{fig:design_concept}(A-F). 
 

For most targeting settings, existing privacy-enhancing technologies (PETs) perform poorly. To illustrate, Figure~\ref{fig:pets_accuracy} shows how two different programs --- an anti-poverty program in Togo (Figure \ref{fig:pets_accuracy}A) and a micro-lending platform in Nigeria   (Figure \ref{fig:pets_accuracy}B) --- are impacted by common PETs. We provide details on both programs below, but at a general level both programs used machine learning, in conjunction with personal data from mobile phones, to determine an individual's eligibility for a valuable benefit. In Togo, where the benefit was an emergency cash transfer, we show results for a program in which 29\% of 4.95 million total individuals were eligible. In Nigeria, where the benefit was a micro-loan, we show results for a program in which low-risk borrowers (i.e., those with a high probability of repaying the loan) are eligible. 
In Figure~\ref{fig:pets_accuracy}, the left-most blue bar indicates baseline performance before any privacy protections are provided, and the two middle red bars indicate performance using differential privacy \cite{dwork2006calibrating,dwork2006our} and $k$-anonymity \cite{sweeney2002k,lefevre2006mondrian}. In the anti-poverty initiative, predictive accuracy falls by $4.7\% - 9.7\%$; for a national anti-poverty program this would translate into $115K - 240K$ additional exclusion errors (i.e., true poor who were excluded from receiving benefits as a result of inaccurate targeting). In the Nigerian consumer lending application, the accuracy of the credit scoring model fell by $16.1\% - 16.4\%$; had the micro-lending platform utilized these PETs, the relative profit of the lending program would have been reduced by $430\% - 476\%$ (see \textit{Methods, Case Studies}).

This stark performance loss arises from a conceptual mismatch between existing PETs and the targeting setting. Accurate targeting relies on an algorithm's ability to correctly distinguish between individuals of different types, based on observed characteristics of those individuals. However, differential privacy is designed to enable statistical learning about \textit{groups} of individuals, while restricting ``excessive'' learning about any individual in the group  \cite{dwork2019differential} (``excessive'' is made precise in \textit{Methods}, Definition 1); this process thus obscures information about individuals that is critical to targeting. $k$-anonymous algorithms, by contrast, are designed to provide ``protection in the crowd'' by ensuring that each individual in the privatized dataset appears identical to at least $k-1$ other individuals \cite{sweeney2002k}; this process thus inhibits a targeting application's ability to differentiate between individuals who have different eligibility statuses, but who are grouped together in the $k$-anonymous algorithm.

This motivates the definition of $(B,\epsilon,\delta)$-\textit{targeted differential privacy} (TDP). At a conceptual level, TDP augments the definition of differential privacy by introducing an auxiliary parameter $B$ to ensure that $X_{priv}$ contains sufficient information to delineate between ``sufficiently different'' individuals, but insufficient information to delineate between ``sufficiently similar'' individuals (see Definition 2 in \textit{Methods}). This enables TDP to interpolate between two extremes: as $B$ increases, the privacy TDP provides approaches that of $(\epsilon, \delta)$-differential privacy; as $B$ decreases, TDP privacy approaches the protection already present in $X$. By construction, TDP algorithms can generate private datasets that enable algorithms to delineate between sufficiently different individuals, thereby permitting accurate targeting outcomes. For more information on contextual adaptations made to the definition of differential privacy, see \textit{Methods} (\textit{Related Contextual Adaptations}).



We formalize the relationship between TDP and accurate targeting in Theorem 1 (see \textit{Supplemental Information}), which has two implications for the theory and practice of targeting using privatized data. First, this theorem demonstrates an inherent limitation on the use of differential privacy in the targeting setting: as we strengthen the privacy parameters of any differentially private algorithm, we  diminish the ability to correctly target individuals. In particular, this theorem shows that results for differential privacy in Figure~\ref{fig:pets_accuracy} are not an accident -- in general, differentially private algorithms cannot reliably generate $X_{priv}$ that provide accurate predictions for targeting. Second, Theorem 1 provides potential values for $(B, \epsilon, \delta)$ to achieve a desired level of targeting accuracy for program officials to consider as they configure TDP algorithms (see Example 1 in \textit{Methods}). Together, the definition of TDP and Theorem 1 provide a framework that makes it possible to translate the needs of program designers into values for $(B,\epsilon,\delta)$ that enable provably private and accurate targeting.

Our second contribution is a novel algorithm that satisfies the definition of TDP. This innovation is required because other common privatization strategies --- beyond differential privacy and $k$-anonymity --- are likewise inappropriate for targeting settings. For instance, we cannot simply delete individuals through subsampling \cite{balle2018privacy}, as this process would mechanically exclude them from eligibility for the benefit. Likewise, it is not viable to insert ``fake'' individuals \cite{balle2018privacy, kohli2018epsilon, mckenna2021winning}, as targeting fake individuals reduces resources available to real individuals. 
Instead, Algorithm \ref{algo:edprp} provides a general-purpose method for constructing privatized datasets satisfying TDP (see \textit{Methods, Private Projection Algorithm}). Figure \ref{fig:design_concept}(G)-(I) provides a conceptual schematic of the process. The algorithm adapts theory of differentially private projections \cite{kenthapadi2012privacy, blocki2012johnson, gondara2020differentially} and Johnson–Lindenstrauss transforms \cite{chen2015johnson, nabil2017random}, and behaves as follows. The records in a original dataset $X$ are randomly mapped to a higher-dimensional space in a manner that satisfies TDP with favorable statistical and scaling properties: as the dimensionality of this space increases, the amount of noise needed to satisfy TDP decreases. To avoid the difficulties that arise in high-dimensional machine learning, the algorithm then projects these datapoints back into the original dimensionality using a two-step process: first, we privately learn the statistical structure in $X$ by computing a TDP covariance matrix; we then use the right singular values of this privatized covariance matrix to return the datapoints to their original space. The resulting set of datapoints ($X_{priv}$) satisfy TDP, and have values that approximately preserve the statistical structure of $X$. This is of central importance for machine learning applications, as similar records in $X$ tend to remain similar in $X_{priv}$, and dissimilar records in $X$ tend to remain dissimilar in $X_{priv}$.

\section*{Evaluation of the privacy-program effectiveness tradeoff}
\label{sec:eval}
Our third set of results empirically characterize the tradeoffs induced by targeted differential privacy (TDP). We perform this analysis using data from the two real-world applications previewed in Figure \ref{fig:pets_accuracy}: an anti-poverty program in Togo and a consumer lending platform in Nigeria. Broadly, we observe that TDP induces a \textit{privacy-accuracy tradeoff} \cite{dwork2019differential, petti2019differential}: as the privacy guarantee becomes stronger, the performance of the downstream application is degraded. However, considerable nuance can be found in the nature of these tradeoffs.

Since privacy can be infringed upon in a multitude of ways \cite{mulligan2016privacy, kozyreva2021public}, our empirical analysis quantifies the extent to which TDP provides protection against three distinct privacy threats that have been of central concern to legislators and privacy professionals. The first are \textit{singling-out attacks}, where an adversary isolates an individual's data in $X$. Singling-out attacks are explicitly mentioned in Europe's flagship data privacy law, the General Data Protection Regulation (GDPR, Article 29 WP 216) \cite{gdpr2014}. The second threat we consider, also referenced in the GDPR, are \textit{attribute inference attacks}, where an adversary infers the fields of $X$. The last threat we consider are \textit{distinguishing attacks}, where an adversary distinguishes between likely and unlikely values in $X$. For all three attacks, we provide a privacy protection score that ranges between 0 and 1, with higher values corresponding to stronger privacy guarantees. For technical details on the attacks and scores, see \textit{Methods} (\textit{Privacy Attacks}).

\subsection*{Humanitarian aid program  in Togo}
\label{togo_results}

We first characterize the privacy-program effectiveness tradeoff in the context of a humanitarian aid program. In 2020, the government of Togo launched the \textit{Novissi} emergency social assistance program to provide cash transfers to needy individuals during the COVID-19 pandemic. When the program was launched in rural areas, the government did not have data to indicate which individuals had the greatest need for assistance; instead, they used a combination of machine learning and personal data. Specifically, working with researchers, they trained a machine learning model to predict each individual mobile subscriber's poverty status using mobile phone metadata obtained from the country's two mobile phone operators. While the poverty scores produced from mobile phone metadata were imperfect, they were significantly more accurate than the other targeting mechanisms available to the government during the COVID-19 pandemic  \cite{aiken2022machine}.

While humanitarian crises may justify the use of personal data for the greater good \cite{rights1996international, eff2014necessary}, there still remains an imperative to protect personal privacy in such applications \cite{oliver2020mobile}. We present results that compare the exclusion errors (i.e., the number of individuals who are truly poor but who are incorrectly excluded from the program) when privatized data are used in place of original data. Figure~\ref{fig:tradeoffs}A illustrates the tradeoff between protection against singling-out attacks (which increase along the $x$-axis)
and program effectiveness (which decreases in the number of exclusion errors, shown on the $y$-axis). The curve illustrates how the tradeoff between program effectiveness and privacy protection varies with different values of $B$, the key tuning parameter for our TDP algorithm. There are three points labeled on the curve: a blue star, representing the non-private status quo (i.e., what was used by the government of Togo at the time); a green star, corresponding to our algorithm when $B = 0.25$; and a red star, corresponding to traditional differential privacy. As $B$ increases, TDP provides increasingly strong privacy protections (singling-out protection approaches 1), consistent with prior mathematical results proving that differential privacy thwarts singling-out attacks (provided the parameters $\epsilon$ and $\delta$ are sufficiently small) \cite{cohen2020towards}. When $B = 0.25$, the algorithm reaches a balance between two extremes: compared to the non-private status quo, our approach incurs a modest increase in exclusion errors in exchange for substantive increases in privacy protection ($2K$ additional exclusion errors vs. 8,480\% increase in singling-out protection); compared to differential privacy, our approach incurs a modest loss in privacy protection in exchange for a significant reduction in exclusion errors (12.6\% decrease in singling-out protection vs. $113K$ fewer exclusion errors). 

Figure~\ref{fig:tradeoffs}B illustrates a similar, albeit more nuanced, tradeoff between program performance and  protection against attribute inference attacks.  When $B = 0.25$, our method provides a substantial increase in privacy protection over the non-private status quo (by 7,300\%); compared to differential privacy, our approach  incurs a reduction in protection of just 3.4\%. Unlike singling-out protection, we find that increasing $B$ does not result in a monotonic increase in attribute inference protection (as evidenced by the jagged shape of the tradeoff curve), and the attribute inference protection never  reaches 1. This arises because differential privacy, as well as targeted differential privacy, use noise to mask identifiable information across individuals; this is distinct from using noise to mask statistical information across attributes, which would thwart an attribute inference attack. Similar results have been observed in the field of adversarial machine learning, where attribute inference attacks have successfully thwarted differentially private machine learning models \cite{jia2018attriguard, jayaraman2019evaluating}.

Figure~\ref{fig:tradeoffs}C illustrates the tradeoff between program effectiveness and protection against distinguishing attacks. Compared to the non-private status quo, the value $B = 0.25$ increases distinguishing protection from 0 to 6.66\%, at the cost of $2K$ additional exclusion errors. Differential privacy provides substantially greater protection against distinguishing attacks (83.5\%), but at the considerable cost of $115K$ additional exclusion errors. These sharper tradeoffs arise from the fact that the curve in Figure~\ref{fig:tradeoffs}C is approximately linear, which in turn arises from Theorem 1, since distinguishing attacks are intimately related to the protection provided by differential privacy \cite{nasr2021adversary}. More generally, the shape of the curve implies that it may be difficult to achieve strong protections against  distinguishing attacks while also allowing for accurate targeting, since accurate targeting requires the preservation of information to distinguish between eligible and ineligible individuals.

Lastly, we note that increased privacy does not always require a reduction in program effectiveness. For instance, when $B < 0.25$, the private algorithm actually achieves marginally higher program effectiveness than when $B = 0$. This finding is consistent with results in the machine learning literature that indicate that the careful introduction of noise can --- in some situations --- improve the predictive accuracy of a model by reducing overfitting \cite{srivastava2014dropout, hardt2016train,ji2014differential}. 



\subsection*{Consumer lending application in Nigeria}
\label{nigeria_results}

Our second empirical example characterizes the privacy-program effectiveness tradeoff in the context of a micro-lending platform in Nigeria. Like many other such platforms that have recently gained widespread popularity in LMICs \cite{francis2017digital}, the platform we study leverages the mobile phone network to distribute and collect payments on very small loans (roughly \$10). A distinguishing feature of these ``digital credit'' providers is that, since many of their customers do not have bank accounts or formal financial histories, they make lending decisions based on alternative credit scores that are derived from customers' digital ``data footprints.'' The insight behind these alternative credit scores is that data passively generated by mobile phone use -- such as the volume of international phone calls made by the customer -- are excellent predictors of loan repayment \cite{bjorkegren2020behavior}. Lenders use machine learning algorithms to detect these patterns, and  offer loans to people with a high probability of repayment. While many are excited by the potential for these products to ``bank the unbanked'' \citep{burgess2005rural}, there are also privacy concerns about how personal data are being used \cite{blumenstock2023big}.

We present results that compare the profitability of the lending program when privatized data are used in place of original data. We evaluate privacy as protection against the three attacks discussed previously, and measure profit as the revenue generated by lending decisions that result in repayment, minus the loss incurred by lending decisions that result in default (see \textit{Methods, Case Studies}).

Broadly, we observe a similar qualitative pattern in the Nigerian lending platform as with the humanitarian aid program in Togo: there are stark tradeoffs to be made at the two extremes of privacy protection, but large gains in privacy can be obtained for modest sacrifices in profit.  The tradeoff between lender profit and singling-out protection is shown in Figure \ref{fig:tradeoffs}D, with three points of interest: a blue star, representing the non-private status quo; a green star, corresponding to our algorithm when $B = 0.1$; and a red star, corresponding to differential privacy. We find that, relative to the non-private approach, our algorithm increases singling-out protection by 7,936\% in exchange for a 12\% reduction in relative profit; relative to the differentially private approach, our approach increases relative profit by 79\%, at the expense of a 11\% loss in privacy protection. Figure \ref{fig:tradeoffs}E displays the relationship between attribute inference inference and the program's profitability. Interestingly, we find that our approach has stronger privacy protection than \textit{both} the non-private approach and the differentially private approach (an increase of 8,788\% and 9\% increase respectively). In particular, increasing $B$ from 0 to 0.05 yields drastic increase attribute inference protection, whereas subsequently increasing $B$ to 0.1 results in a negligible increase in protection. In contrast to the prior two attacks, Figure \ref{fig:tradeoffs}F displays a sharp tradeoff between distinguishing protection and the program's profitability, consistent with the results of Theorem 1 that accurate targeting necessarily requires smaller values of $B$.

\subsection*{Navigating tradeoffs in practice}

While the broad takeaway from both case studies is the same --- that there generally exists a tradeoff between privacy and program effectiveness, and that large gains in one can be obtained for small reductions in the other --- there are nuanced differences between the cases that highlight the sort of tradeoffs that program officials may need to make in practice. Across both programs, for instance, increasing $B$ results in a monotonic increase in protection against singling-out and distinguishing attacks. However, in both programs we find that attribute inference protection increases very rapidly up to a certain point; beyond that point, little protection is gained by increasing $B$. Interestingly, the protection is neither monotonic in $B$, nor does it approach 1 as we increase $B$. In fact, our results in Nigeria demonstrate that differential privacy actually exhibits \textit{lower} attribute inference protection than targeted differential privacy when $B = 0.1$. 

Taken together, these results suggest that there may be more than just a privacy-accuracy and privacy-program effectiveness tradeoff; in some situations, there may be a \textit{privacy-privacy tradeoff}, where, for instance, increases in singling-out protection are achieved through decreases in attribute inference protection (and vice-versa). When such situations arise, data holders and program officials must collaborate to decide which types of privacy protections are most important to prioritize, subject to accuracy constraints. In the two empirical settings we study, singling-out and attribution inference protections are not as costly to provide, so they may be easier to prioritize. However, such decisions imply a different set of societal values and priorities, and are deeply contextual \citep{nissenbaum2004privacy}. The decisions must account for several factors, including the contents of the specific data in question, relevant privacy laws and international humanitarian standards, and the feasibility of deploying additional privacy-monitoring systems (e.g., query logging can be used to record the computations executed on the privatized data; while imperfect \cite{kohli2023differential}, logging can be used to monitor downstream users for adversarial behavior and reduce privacy risk \cite{bowman2015architecture, kroll2019privacy}).


\section*{Discussion}
\label{sec:disc}
This paper develops a framework and efficient algorithm for targeted differential privacy: a paradigm for constructing provably private datasets that still enable applications that make targeting decisions about individuals. As illustrated by the case studies in Togo and Nigeria, those privacy guarantees can impact program effectiveness. While the exact nature of that tradeoff is context-dependent, we consistently observe that large gains in privacy can be obtained for small reductions in program effectiveness.

The insights from Togo and Nigeria are likely to generalize to a wide range of settings where personal data are used to determine eligibility for policies, products, and benefits. We focus on two low- and middle-income countries (LMICs) because these are settings where, until recently, most individuals were not generating rich ``digital footprint'' data that could be re-analyzed by companies and policymakers. Over the past few years, however, a wide range of LMIC policies and products now rely on sensitive personal data to make high-stakes decisions --- including in settings of food security, financial services, and social protection  \cite{blumenstock2023big, mumtaz2021machine, aiken2023estimating, kirkpatrick2023using}. While our focus on LMICs is thus not accidental, we expect that our framework and algorithm for targeted differential privacy could likewise be applied in high-income settings. 




As we consider how privacy protections can be provided for other applications, more work is required to characterize the tradeoff between privacy and program effectiveness in each new setting. This is will require deep collaboration between data holders, program designers, and the research community. However, the technical approach we propose should not be mistaken as a substitute for the privacy protections afforded by privacy laws, policies, and practices --- rather, it complements existing protection mechanisms by providing assurances about a privatized dataset's resilience to technical data privacy attacks. We hope future research will build upon our work, and generate a deeper understanding of the real-world tradeoffs that arise when targeted differential privacy (and privacy protections more broadly) is used to make targeting decisions. In so doing, this can facilitate new opportunities for data access and data sharing, and help distill the complexities of data privacy into a collection of legible and measurable tradeoffs.

\newpage
\section*{Methods}
This section provides a brief description of the mathematical foundations of targeted differential privacy, our private projection algorithm, and the experimental methods and results. Full details of our approach, including formal definitions and proofs of the theorems referenced, can be found in the \textit{Supplemental Information}.



\subsection*{Targeted Differential Privacy}

\paragraph{Preliminaries.} Differential privacy was first introduced in 2006 to enable the computation of privacy preserving statistics \cite{dwork2006calibrating,dwork2019differential}. At the root of differential privacy is a normative distinction about what constitutes a privacy violation in statistical learning \cite{kohli2021leveraging}: inferring information about a population is not considered a privacy violation, provided the statistical analysis does not reveal ``too much'' about any individual.

Differential privacy achieves this through the concept of neighboring datasets. Let $||\cdot||_2$ denote the $L_2$-norm, and $\L_2^{n \times d}$ denote the set of of real $n \times d$ matrices where the $L_2$ norm of each row is at most 1. We say two databases $X, X' \in \L_2^{n\times d}$ are \textit{classic neighbors} if they agree on exactly $n-1$ rows. Conceptually, a randomized algorithm is differentially private (DP) if the probability of creating a privatized dataset $X_{priv}$ is essentially the same (up to $\epsilon$ and $\delta$) for every pair of neighboring datasets. 

\begin{defn}
    A randomized algorithm $A$ from $\L_2^{n\times d}$ to $\mathbb{O}$ satisfies $(\epsilon,\delta)$-\textit{differential privacy} if for all classic neighbors $X,X' \in \L_2^{n\times d}$ and for all events $E \in \mathbb{O}$, $\P(A(X) \in E) \le e^{\epsilon}\P(A(X') \in E) + \delta$.
\end{defn}

Targeting applications are different from population-level analyses since the ability to target individuals involves determining their eligibility status. This suggests that we must adapt the notion of neighboring datasets to preserve sufficient information in $X_{priv}$ to differentiate between different eligibility statuses. We say  $X, X' \in \L_2^{n\times d}$ are \textit{$B$-neighbors} if they are classic neighbors, and for the single row $i$ where they disagree, $||X_i - X'_i||_2 \le B$. When $B = 2$, $B$-neighbors coincides with classic neighbors. This motivates the definition of targeted differential privacy (TDP) below, which quantifies protection over $B$-neighboring datasets in place of classic neighboring datasets. 

\begin{defn}
    A randomized algorithm $A$ from $\L_2^{n\times d}$ to $\mathbb{O}$ satisfies $(B, \epsilon,\delta)$-\textit{targeted differential privacy} if for all $B$-neighbors $X,X' \in \L_2^{n\times d}$ and for all events $E \in \mathbb{O}$, $\P(A(X) \in E) \le e^{\epsilon}\P(A(X') \in E) + \delta$. By definition, when $B = 2$ we recover $(\epsilon, \delta)$-differential privacy. 
\end{defn}

Through this alteration, the auxiliary parameter $B$ ensures $X_{priv}$ contains sufficient information to distinguish different individuals (i.e., those whose data are more than $B$ apart), but insufficient information to learn the specific data values of similar individuals (those whose data are at most $B$ apart).

\paragraph{Related contextual adaptations.}  
There is a rich scholarly history of altering the definition of differential privacy to better align with context-specific goals outside of statistical domains. For example, the concept joint differential privacy has been applied matching and allocation problems to cope with the fact that these problems are provable impossible to solve under differential privacy \cite{kearns2014mechanism, hsu2014private}. In another example, a variant of differential privacy was introduced to enable the aggregation of privacy preferences in a self-referential voting system \cite{kohli2018epsilon}. And, the concept of geo-indistinguishability was introduced to adapt differential privacy to location-based systems \cite{andres2013geo}. In all of these cases, the concept of differential privacy was altered in a principled manner to be congruous with the context-specific goals of the task at hand.

The definition of targeted differential privacy follows in this tradition by using the notion of $B$-neighbors to better account for the context-specific goals of targeting. This notion of $B$-neighbors has been used previously in several settings unrelated to targeting \cite{desfontaines2019sok}: in spatial data analysis, this requirement is called differential privacy under a neighborhood \cite{fang2014differential}; in privacy for compressed sensing applications, this definition is called constrained differential privacy \cite{zhou2009differential}; and in privacy for moving-horizon estimators, this definition is referred to as adjacent differential privacy \cite{desfontaines2019sok, krishnan2020probabilistic}. The definition of targeted differential privacy is also similar in spirit to metric differential privacy \cite{chatzikokolakis2013broadening}, which is an alternative generalization of classic differential privacy for contexts where neighboring definitions are ill-suited (such as geolocation and smart-metering tasks). We differentiate from prior work by showing how the parameter $B$ affects and enables targeting applications below. 


\paragraph{Necessary conditions for accurate targeting.} Using $(B,\epsilon,\delta)$-targeted differential privacy, we can formalize the tradeoffs that manifest in all targeting applications. To enable accurate targeting, we need a targeting process to behave ``similarly'' when $X_{priv}$ is used in place of $X$. This is formalized by assigning sufficient probability $\gamma$ to the event that an individual's eligibility status remains the unchanged when privatized data is used in place of original data. Theorem 1 in \textit{Supplemental Information} proves that if this accuracy conditions holds, then $B, \epsilon$, and $\delta$ must exhibit the following relationship:
$$ 
    \ceil{2B^{-1}} \ge \ceil{\epsilon^{-1}\ln(Q)}
$$
where 
$$
Q = \frac{\delta + \gamma(e^{\epsilon}-1)}{\delta + (1-\gamma)(e^{\epsilon}-1)}
$$ 
This theorem provides the necessary conditions that \textit{every} targeted (and classic) differential privacy algorithm must obey in order to enable accurate targeting. This has two implications for the theory and practice of provably private targeting. 

The first implication is that accurate targeting (as specified by $\gamma$) is provably impossible for classic differential privacy with small parameter values. To unpack this relationship, consider the case when $\delta = 0$. Then, $Q = \gamma /(1-\gamma) $, so the inequality above reduces to $\ceil{2B^{-1}} \ge  \ceil{\epsilon^{-1}\ln(\gamma/(1-\gamma))}$. As we strengthen the privacy parameter $\epsilon$ to the point where $\epsilon < \ln(\gamma/(1-\gamma))$, $\ceil{\epsilon^{-1}\ln(\gamma/(1-\gamma))} \ge 2$; hence, $\ceil{2B^{-1}} \ge 2$, yielding $B < 2$. A similar result also holds when $\delta > 0$ (see Example 1 below). Thus, when $\epsilon$ and $\delta$ are small, classic differential privacy cannot permit accurate targeting.

Second, this result can be used proactively to provide a collection of potential candidate values for $B, \epsilon,$ and $\delta$ for \textit{any} targeted differential privacy algorithm based on the accuracy needs of a situation. This follows by the contrapositive of Theorem 1: if $(B,\epsilon,\delta)$ satisfy $\ceil{2B^{-1}} 
 < {\ceil{\epsilon^{-1}\ln(Q)}}$, then the accuracy bound does not hold. We present one such example below. 

\begin{ex}   
\label{ex:togo_params} A program official in Togo wants to use a TDP algorithm for their anti-poverty program with a value of $B$ such that $2B^{-1} \in \Z$. To ensure humanitarian relief is correctly dispersed, they require $99\%$ confidence that an individual's eligibility status will remain unchanged when $X_{priv}$ is in place of the original data $X$. For such $B$ and $\gamma$, the bound from Theorem 1 simplifies to 
$$
B \le \frac{2}{\ceil{\epsilon^{-1}\ln(Q)}}
$$
If the program official uses $\epsilon = 1$ and $\delta = 1\times 10^{-4}$, then by they only need to consider $B \le 0.4$, since no TDP algorithm exists to meet the accuracy bound with $B > 0.4$. Alternatively, if program officials are willing to use $\epsilon = 4$, then they only need to consider $B \le 1$.
\end{ex}

In practice, the parameter $B$ required to enable accurate targeting may need to be smaller than those produced by the bound in Theorem 1. Semantically speaking, this follows as Theorem 1 provides the \textit{necessary} conditions for $(B, \epsilon, \delta)$ to meet a desired accuracy bound, but not \textit{sufficient} conditions. In general, the extent to which $B$ must be reduced to invoke a sufficient condition for accurate targeting is task specific, and will depend on the underlying relationship between the original data's features and the targeting outcome in question, as well as the specifications of the targeting process and the specific TDP algorithm used.

\subsection*{Private Projection Algorithm}  
Our private projection algorithm (Algorithm \ref{algo:edprp}) requires 7 inputs: the dataset $X$ that we wish to privatize, 5 privacy parameters $(B, \epsilon_1, \epsilon_2, \delta_1, \delta_2)$, and a parameter $k$ that represents the dimensionality of a mathematical projection used internally by our algorithm. The strategy of the algorithm is based on Gondara and Wong's random projection method \cite{gondara2020differentially}, with three technical alternations (detailed in \textit{Supplemental Information}). Our algorithm begins in Steps 1-3 by randomly map datapoints in $X$ from $\R^d$ to $\R^k$ in a manner that satisfy $(B,\epsilon_1, \delta_1)$-TDP (Proposition 2 in \textit{Supplemental Information}). In Steps 4-6, the algorithms map these points back to $\R^d$ using the data's covariance structure in a manner that satisfy $(B,\epsilon_2, \delta_2)$-TDP (Lemma 7 in \textit{Supplemental Information}). Due to this modular design, our private projection algorithm satisfies $(B, \epsilon_1 + \epsilon_2, \delta_1+\delta_2)$-targeted differential privacy (Theorem 2 of \textit{Supplemental Information}).

\begin{algorithm}
\KwIn{
\begin{itemize}
    \item Dataset $X$
    \item Privacy parameters $B \in (0,2], \epsilon_1 > 0, \epsilon_2 \in (0,1), \delta_1 \in (0,1), \delta_2 \in (0,1)$
    \item Projection parameter $k$
\end{itemize}
}
\KwOut{Privatized matrix $X_{priv}$}
Sample $R \in \R^{d \times k}$, where each entry is sampled uniformly at random from $\{-1,0,1\}$. \\
Compute the projection $P = k^{-1}XR$ \\
Privatize the projection via $P_{priv} = P + G$ where $G \in \R^{n \times k}$, $G_{i,j} \sim N(0,\sigma^2)$, and $$\sigma = \frac{B}{\sqrt{k}}\sqrt{d\ln((2/3)(e-1) + 1) -k^{-1}\ln(\delta_1/2)} \frac{\sqrt{2(\ln(1/\delta_1) + \epsilon_1)}}{\epsilon_1} $$\\
Compute the privatized covariance matrix of $C_{priv} = X^T X + G$ where $G_{i,j} \sim N(0, \sigma^2)$ and $$\sigma = 2B\sqrt{2\ln(1.25/\delta_2)}/\epsilon_2$$ for all $i \ge j$, and $G_{i,j} =  G_{j,i}$ for all $i < j$.\\
Compute the right singular vectors of $C_{priv}$. Call them $V^{T}$. \\
\textbf{Return} $X_{priv} = P_{priv}(V^{T}R)^{\dagger}V^{T}$, where $\dagger$ represents the pseudo-inverse.
\caption{Private Projection Algorithm}
\label{algo:edprp}
\end{algorithm}


\medskip

\noindent \textbf{Constructing TDP (and DP) datasets.} We experimentally evaluate the impact our private projection algorithm induces when $X_{priv}$ is used for targeting in place of $X$. In both of our case studies, we normalize each row of $X$ to ensure the $L_2$ norm is at most 1. Hence, when $B=2$ our algorithm generates $X_{priv}$ with classic differential privacy guarantees. For our simulations, we define the \textit{parameter grid} of our experiments as all possible 6-tuples $(B, \epsilon_1,\delta_1,\epsilon_2,\delta_2,k)$, where
\begin{itemize}
    \item $B \in \{0.05, 0.075, 0.1, 0.25, 0.5, 0.75, 1, 2\}$
    \item $\epsilon_1 \in \{2, 3\}$
    \item $\delta_1 = \frac{2}{3}\delta$
    \item $\epsilon_2 \in \{0.5, 0.9999\}$
    \item $\delta_2 = \frac{1}{3}\delta$
    \item $k = 10^4$
\end{itemize}

\noindent and the specific value for $\delta$ used to compute $\delta_1$ and $\delta_2$ is determined by the number of individuals in each of the case studies. As a technical note, we use 0.9999 instead of 1 for the value of $\epsilon_2$, as the Gaussian perturbation of Step 4 requires $\epsilon_2$ in (0,1) \cite{dwork2006our, balle2018privacy, dwork2014analyze} (see Lemma 7 of \textit{Supplemental Information}). 



Following best practices of statistical agencies that utilize differential privacy, we set $\delta$ to be less than the inverse of the number of individuals whose data is being privatized by an algorithm \cite{page2018differential, abowd20222020}. In the Togolese anti-poverty experiments, we privatize the data all at once with our algorithm, and hence set $\delta = (n+1)^{-1}$, where $n$ is the number of individuals in the Togolese dataset. In the Nigerian micro-lending experiments, we construct $X_{priv}$ by first randomly partitioning the original dataset $X$ into 6 sub-matrices of nearly equal size; we then apply the private projection algorithm to each sub-matrix, and combine the privatized sub-matrices together to form $X_{priv}$. Since our private projection algorithm is applied to disjoint datasets of size at most $\ceil{n/6}$ (where $n$ is the number of individuals in the Nigerian dataset), we set $\delta_1$ and $\delta_2$ using the value of $\delta = (\ceil{n/6} + 1)^{-1}$ when we privatize each sub-matrix. By Lemma 3 in \textit{Supplemental Information}, this parallelization strategy does not increase the privacy loss of our algorithm. 

This partitioning approach improves both the runtime and accuracy of our algorithm. Parallelization removes computational barriers that arise when performing matrix operations on large datasets (especially as $n,d,$ and $k$ increase). Additionally, the parallel computation of privatized disjoint sub-matrices of $X$ improves the accuracy of $X_{priv}$ by decreasing the amount of noise introduced in Steps 3 and 4 of the algorithm (as $\delta_1$ and $\delta_2$ are inversely related to $n$).

\medskip



\subsubsection*{$k$-anonymity baseline} As an additional point of comparison, we examine the program effectiveness in our two case studies when $X_{priv}$ satisfies $k$-anonymity \cite{sweeney2002k}, another commonly applied privacy standard. We utilize the Mondrian algorithm \cite{lefevre2006mondrian} to construct a $k$-anonymous $X_{priv}$, which has been empirically shown by Slijepvcevic et al. (2021) to outperform other $k$-anonymous algorithms for machine learning tasks \cite{slijepvcevic2021k}. To preserve experimental consistency across the research literature,  we use the same open-source Mondrian implementation as Slijepvcevic et al. (2021), which requires three inputs: the original dataset $X$, the privacy parameter $k$ (not to be confused with the value of $k$ in the private projection algorithm), and a list of \textit{quasi-identifiers} in $X$ (i.e., the names of the columns that need to be privatized). Since our private projection algorithm introduces noise to every column of $X$, we configure Mondrian to treat every column of $X$ as a quasi-identifier. This enables a baseline level of consistency between the targeted (and classic) differential privacy experiments and $k$-anonymity experiments in the sense that every column of $X$ is subject to privacy alterations. In all of our experiments, we consider values of $k \in \{2,...,10\}$.

\medskip





\subsection*{Case Studies}

\subsubsection*{Anti-poverty program in Togo}

Following in Aiken et al. \cite{aiken2022machine}, we use a machine learning model to determine eligibility for Togo's Novissi (anti-poverty) cash transfer program based on mobile phone metadata. The algorithm is trained on a dataset that matches individual survey responses to data obtained from mobile phone operators in Togo. The survey makes it possible to observe the average daily consumption (per capita, purchasing-power-parity adjusted) of a nationally representative sample of roughly 4,200 mobile subscribers in Togo. We refer to this measure of consumption as our target variable $y$, which we will try to predict from data on how those subscribers use their mobile phones.  The second dataset, $X$, obtained from Togo's two mobile phone operators, contains 10 ``features'' that quantify the phone use of each of the mobile phone subscribers in survey, such as the number of nighttime calls made, the entropy of the number of contacts they text during the weekday, and the percent of calls in specific prefectures.  We normalize each column of this dataset of features $X$ to have mean 0 and standard deviation 1, and normalize each row to be in the 10-dimensional $L_2$ ball. The target variable is normalized to be in the $[0,1]$ interval. Additional statistical information on the features and targets can be found in Table S1 in \textit{Supplemental Information}.

We use 5-fold cross-validation to train a ridge regression model to predict consumption $y$ from mobile phone use $X$. To simulate the targeting of Novissi, we label an individual as eligible for the Novissi program if their predicted  consumption is below the $29^{th}$ percentile of all predictions (this corresponds to the budget constraint of the actual program implemented by the government). For each fold, we record the average accuracy and average false positive rate. Using these estimates, we estimate the average number of ``true poor'' individuals (i.e., those whose actual consumption is below the $29^{th}$ percentile) who are incorrectly excluded from the program because their predicted consumption is above the $29^{th}$ percentile. To generate estimates of exclusion errors that are nationally representative, we scale the exclusion errors from the surveyed sample of 4,200 to Togo's estimated adult population (i.e., individuals at least 15 years of age) of approximately 4.95 million individuals. (According to the World Bank, Togo's estimated population in 2019 was 8,243,094 \cite{wbTogo2019}; and based on estimates from the United Nations Population Fund, approximately 60\% of individuals in Togo are at least 15 years of age \cite{unpfTogo}). 




To show how targeted differential privacy affects the performance of this program, we simulate how targeting would have worked if $X_{priv}$ had been used in place of $X$. Specifically, for each 6-tuple of the privacy parameters $(B,\epsilon_1,\delta_1,\epsilon_2,\delta_2,k)$ we perform 50 simulations whereby we generate $X_{priv}$ using our private projection algorithm, and then use the same 5-fold cross-validation splits to compute the average accuracy and average false positive rate of the targeting procedure. From these estimates, we compute the average number of true poor excluded from the program at the national scale. 

To simulate program performance under $k$-anonymity, we construct $X_{priv}$ using the Mondrian algorithm for each privacy parameter $k \in \{2,...,10\}$. We then use the 5-fold cross validation process to compute the average accuracy and average false positive rate of the targeting procedure. Since Mondrian's algorithm does not utilize randomness, we only generate one $X_{priv}$ for each value of $k$. Using these estimates, we compute the average number of true poor excluded from the program at the national scale, and report the results for the most accurate parameter value ($k=2$).




\subsubsection*{Micro-lending platform in Nigeria}

Our analysis of the privacy-program effectiveness tradeoff in the Nigerian micro-lending case study relies on three data sources. The first dataset $X$ consists of 15 features of 20,788 individuals, which quantify how each individual uses their phone. The second dataset $y$ characterizes individuals as low-risk or high-risk borrowers, as determined by an alternative credit score derived by the Nigerian lender. 
Our third dataset contains information on the size of the first loan an individual applied for, which we use to determine the profit or loss associated with repayment or default on the loan. All three datasets contain a hashed pseudonymous identifier that enables us to associate records across datasets. We normalize $X$ and $y$ in the same manner as the Togolese anti-poverty case study. Additional statistical information on the features and credit scores can be found in Table S1 in \textit{Supplemental Information}.

We train a logistic regression model (using 5-fold cross validation) to predict an applicant's riskiness (and hence their eligibility for a loan) using the 15 features derived from their mobile phone metadata. 
To calculate the profit associated with lending decisions based on the privatized versus the original data, we assume that profit accrues to the lender when they offer loans to low-risk borrowers, and decreases when they offer loans to high-risk borrowers. Specifically, for a loan applicant $j$, let $l_j$ denote the requested loan amount, $r_j$ denote the lender's revenue on the loan, and $i_j$ denote the interest on borrower $j$'s loan. Denote the lender's profit from individual $j$ as $\pi_j$, which is determined by the following four cases:%
\begin{enumerate}
    \item If the loan eligibility algorithm correctly classifies $j$ as high risk, then the lender does not offer $j$ a loan. In this case, the lender earns no profit, so $\pi_j = 0$.
    \item If the loan eligibility algorithm correctly classifies $j$ as a low risk, then the lender offers $j$ a loan and $j$ repays the entire amount. In this case, the lender's profit is the revenue generated by the loan, so $\pi_j = r_j$.
    \item If the loan eligibility algorithm incorrectly classifies $j$ as a low risk, then the lender offers $j$ the loan but is unable to collect payment. For our analysis, we assume the worst-case outcome for the lender in which $j$ repays none of the loan, and the lender loses the initial value of the loan, plus any interest they would have earned on they loan amount had they not extended the loan. Thus, $\pi_j = - (l_j + i_j)$.
    \item If the loan eligibility algorithm incorrectly classifies $j$ as a high risk, then the lender does not offer $j$ a loan. In this setting, the lender does not receive the revenue they would have earned if they had offered the loan to the low risk applicant, and $\pi_j = -r_j$.
\end{enumerate}

The program's profitability is given by $\pi = \sum_{j \in [n]} \pi_j$. Across all folds, we compute the program's average classification accuracy and the profitability. 

To quantify the impact of using a privatized dataset $X_{priv}$ in place of $X$, we follow the same experimental process used in the Togolese experiments with one notable alteration: to compute estimates for the targeted and classic differential privacy experiments, we use the partitioning strategy described in \textit{Constructing TDP (and DP) datasets} to improve the runtime and accuracy of private projection algorithm.


\subsection*{Privacy Attacks}

We measure the privacy afforded by our algorithm against three threats: singling-out attacks, attribute inference attacks, and distinguishing attacks. Each of these privacy measures --- described in more detail below --- quantifies an adversary's inability to infer information about the original dataset $X$. These measures range between 0 and 1, with higher values corresponding to higher levels of privacy protection. We do not examine the privacy protection afforded by the $k$-anonymous Mondrian algorithm, as the targeting accuracy it exhibited was too stark to be of use in practice (see Figures \ref{fig:pets_accuracy}(A) and (B)). 

In all of our attacks, we assume the adversary has full access to the privatization algorithm, and can hence use any information about the algorithm as part of their attack strategy. This modeling choice is analogous to Kerckhoffs' Principle from cryptography, which states that ``the cipher method must not be required to be secret, and it must be able to fall into the hands of the enemy without inconvenience'' \cite{katz2007introduction}.
This is because cryptographic methods that depend on secrecy can be easily compromised if that secrecy is violated. As such, to illustrate the worse-case privacy scenario, we conduct our privacy attacks with full knowledge of the privatization algorithm to provide a realistic testing scenario. 

\paragraph{Singling-out protection.} Singling-out protection quantifies $X_{priv}$'s resilience to \textit{isolation} (i.e., inferring the existence of a unique row in $X$ containing a certain collection of values). As defined by Article 29 Working Party 216 of the GDPR, singling-out ``corresponds to the possibility to isolate some or all records which identify an individual in the dataset'' \cite{gdpr2014}. We follow the approach of \textit{predicate singling-out} pioneered by Cohen and Nissim \cite{cohen2020towards}: using both the privatized dataset $X_{priv}$ and the privatization algorithm $A$, can an adversary infer a set of conditions that is satisfied by exactly 1 row in $X$? If so, then the adversary has the means to isolate a single row in $X$ (and hence single out this individual). We define the singling-out protection of a dataset as the proportion of individuals we are unable to single-out.  

We quantify the singling-out protection present in the original and privatized datasets in both of our case studies using the following methodology. In the original setting, the data $X$ is unaltered, i.e., $X_{priv} = X$. Hence, an adversary can trivially single-out an individual by determining which rows in $X_{priv}$ are unique. That is, for every row in $X_{priv}$, an adversary can construct the \textit{identity conditions for row $i$} via ``$x_i = (X_{priv})_i$ for all features $i$.'' Since $X_{priv} = X$, the identity conditions for row $i$ singles-out a row in $X$ if and only if $(X_{priv})_i$ is unique in $X_{priv}$. Therefore, the original dataset's singling-out protection score is given by the fraction of non-unique rows in $X$. 


In contrast, when $X_{priv}$ is given by our private projection algorithm, the presence of noise ensures that the identity conditions will almost surely fail to isolate a row. For this reason, we develop a new singling-out attack that leverages the design of our algorithm (called a \textit{net attack}) that generate better predicates by modifying the identify conditions to account for the presence of noise in $X_{priv}$. Let the \textit{$d$-dimensional net of size $\eta \in \R^d$}, denoted as $\text{Net}(\eta; X_{priv})$, be the set of $z \in \R^d$ such that each entry $z_j$ is within distance $\eta_j$ of any row in $X_{priv}$. Then a row $X_i$ is singled-out by the net attack  if and only if $X_i \in \text{Net}(\eta; X_{priv})$ and $|\text{Net}(\eta;X_{priv})| = 1$.

We generate candidate values of $\eta$ based the standard deviation in each column of $X_{priv}$. The rationale for this approach stems from the following two observations: since the standard deviation measures the typical spread of datapoints in each column, the adversary could leverage some constant multiple of the standard deviation to generate possible net sizes; and since our algorithm introduces noise, the standard of each privatized column is likely larger than the standard of each original column. 

We use five different constants in $ \big\{\frac{1}{10}, \frac{1}{3}, \frac{1}{2}, \frac{2}{3}, 1\big\}$ to scale the column-wise standard deviations of $X_{priv}$ to generate five different candidate values for $\eta$. For each parameter tuple in the parameter grid, we run our private projection algorithm 50 times and execute the net attack on each privatized dataset generated using each of the candidate values for $\eta$. For each $\eta$, we compute the average singling-out protection. To appropriately measure the privacy protection afforded to an individual, we report the protection afforded based on the adversary's most successful attack. Hence, we compute $X_{priv}$'s singling-out protection for a parameter tuple as the worst-case protection across all values of $\eta$.

\paragraph{Attribute inference protection.} In contrast to singling-out protection, attribute inference protection quantifies $X_{priv}$'s resilience to feature reconstruction (i.e., inferring column values of $X$). As conceptualized in Article 29 Working Paper 216 of the GDPR, attribute inference refers to ``the possibility to deduce, with significant probability, the value of an attribute from the values of a set of other attributes'' \cite{gdpr2014}. Our framework of attribute inference protection follows the approach described in Giomi et al. \cite{giomi2022unified}, where an adversary has access to the privatized data $X_{priv}$, the privatization algorithm $A$, and $h$ columns of the original data $X$ on all $n$ individuals; using these three objects, the adversary's goal is to infer the missing $d-h$ columns of $X$.

For our analysis, we consider $h \in \{1, \ceil{d/2}, d-1\}$, where the columns available to the adversary are randomly selected when $h \ne d-1$. For our attack method, we use the inference attack module in Anonymeter \cite{giomi2022unified}. This attack method is similar to the approach pioneered by Narayanan and Shmatikov  \cite{narayanan2008robust}, and proceeds as follows. Using the $h$ columns in the original dataset, the adversary performs a nearest neighbor search on the same $h$ columns of the privatized data. Once this record is found in the privatized dataset, the adversary uses $d-h$ values in the privatized data as their guess for the $d-h$ values in the original data. We say the attack succeeds for individual $i$ on unknown column $j$ if the relative error between the estimate and the actual value is at most 5\%, and define the protection score of a column as the proportion of individuals for which the attack fails. We follow the holdout-approach of Giomi et al. \cite{giomi2022unified} and transform this protection score into a relative protection score: since $X$ and $X_{priv}$ share many statistical properties, inferences due to the preservation of statistical structure between their columns are qualitatively different than those that arise due to $X_{priv}$'s retention of precise information on individuals \cite{DPorg-inference-is-not-a-privacy-violation}; as such, this transformed score refines the original score by quantifying $X_{priv}$'s inference risk that cannot be explained by its statistical similarity to $X$ (see \textit{Supplemental Information} for technical details). 

To determine the baseline level of protection in the original dataset, we use the following methodology. When $h \ne d-1$, for every column of $X$ we run Anonymeter 50 times, randomly selecting $h$ other columns during each iteration, and compute the average relative protection score. When $h = d-1$, for every column of $X$ we run Anonymeter once (as there is no randomness involved in the selection of the other columns) and compute the relative protection score.  To assess the protection of our privatized datasets, for each 6-tuple of algorithm parameters in the parameter grid, we replicate the above process 50 times (even when $h = d-1$) to capture the sampling variability of Algorithm \ref{algo:edprp}. We define the attribute inference protection of the original and privatized datasets to the lowest average protection score across all values of $h$ and all columns $j$ (i.e., protection against the most successful attack). 

\paragraph{Distinguishing protection.} Distinguishing protection is based on the mathematical construct of the \textit{privacy loss random variable}, which quantifies an adversary's inability to infer which potential datasets likely produced $X_{priv}$. This protection is modeled by the following cryptographic game \cite{nasr2021adversary}. An adversary has access to $X_{priv}$ and the algorithm $A$ that generated it. Given any two datasets $\hat{X}$ and $\tilde{X}$ that differ in exactly one row, the adversary attempts to determine which dataset was more likely to create $X_{priv}$. The information-theoretic extent to which an adversary can do so is quantified by the privacy loss random variable. In the \textit{Supplemental Information}, we construct our measure of distinguishing protection, which is proportional to the inverse of the mean of the privacy loss random variable. By definition, the distinguishing protection of the original dataset is 0. We derive the distinguishing protection formula for our private projection algorithm, and compute its value for every parameter tuple in our parameter grid. 

\paragraph{Interpreting privacy protection scores.}

The first two measures of privacy protection are computed via privacy audits \cite{cummings2024advancing}, where we attack datasets to determine their resilience against singling-out and attribute inference attacks. A methodological limitation of privacy audits stems from their empirical nature; they can only be used to surface examples of privacy risk, but cannot be used to show an approach is risk-free \cite{kohli2021leveraging, cummings2024advancing}. 
For this reason, the privacy protection scores for singling-out and attribute inference protection should be interpreted as an upper bound on the actual privacy provided by $X_{priv}$. Despite this limitation, privacy audits can still be used to provide complementary measures of privacy protection for formal approaches (such as TDP and DP, whose protection is quantified via the parameters $B, \epsilon,$ and $\delta$) to understand the privacy afforded against real-world attacks \cite{cummings2024advancing}.  Distinguishing protection, in contrast, is computed via the mathematical properties of our algorithm, and provides a mathematical lower bound on the protection afforded.


\subsection*{Data availability}
The mobile phone datasets from Togo and Nigeria contain detailed information on billions of 
mobile phone transactions. These data contain proprietary and confidential information belonging to a private telecommunications operator and cannot be publicly released. Upon reasonable request, we can provide information to accredited academic researchers about how to request the proprietary data from the telecommunications operator.  Contact the corresponding author for any such requests.

\newpage
\bibliographystyle{ieeetr}
\bibliography{references}


\section*{Acknowledgements}

This work was supported by the Lab for Inclusive FinTech (LIFT) at the University of California, Berkeley; the Bill and Melinda Gates Foundation via the Center for Effective Global Action; and the National Science Foundation under CAREER Grant IIS-1942702. The views, opinions, and/or findings expressed are those of the author(s) and should not be interpreted as representing the official views or policies of the U.S. Government. We are grateful to Emily Aiken, Suraj Nair, Dan Cassara, Paul Laskowski, and Seth Garz for providing early feedback on this work.

\section*{Author contributions statement}


JB conceived of the research idea. NK developed the algorithm, proved the theorems, and conducted the empirical measurements. JB and NK wrote the manuscript.

\section*{Competing interests}

The author(s) declare no competing interests.

\section*{Materials and correspondence} 

Correspondence and requests for materials should be addressed to Joshua E. Blumenstock.


\newpage
\section*{Figures}

\begin{figure}[!ht]
\centering
\includegraphics[width=\textwidth]{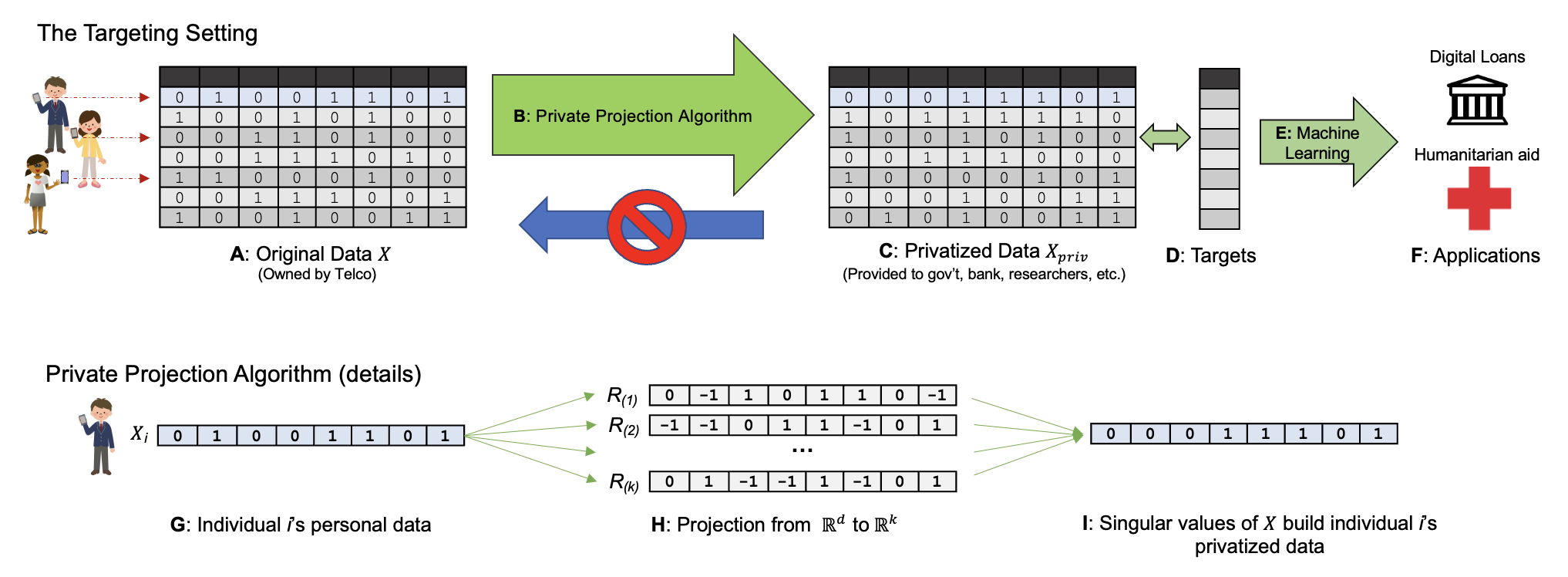}
\caption{Overview of targeting setting (top) and the private projection algorithm (bottom). (A) Personal data are held by a data holder (e.g., mobile network operator). These data are given to (B) an algorithm that generates a version of the data with a provable privacy guarantee. (C) This private version of the data can then be joined with (D) training ``labels'' from third parties that indicate the true eligibility for a subset of the population. (E) Machine learning models learn how to predict eligibility status for the full population for whom eligibility status is not directly observed, but for whom private data exist. (F) These predictions can then be used in downstream applications.  Bottom figures provide a conceptual sketch of how the private projection algorithm works. (G) Each individual's raw data $X$ is projected into $\R^k$ (H), and is then  projected back to $\R^d$ using the singular values of $X$ (I). The resulting record corresponds to individual $i$'s record in $X_{priv}$.}
\label{fig:design_concept}
\end{figure}

\begin{figure}[!p]
\centering
\includegraphics[width=\textwidth]{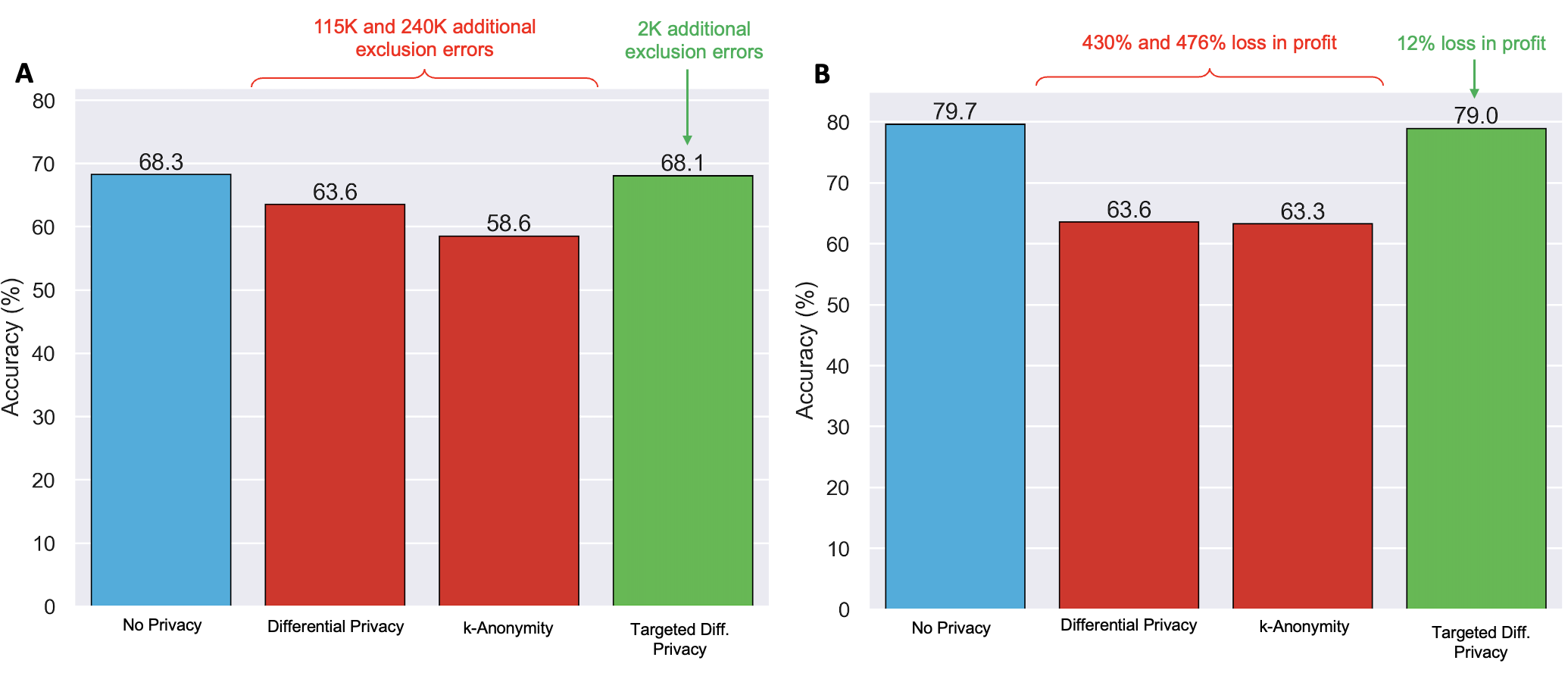}
\caption{The impact of different privacy-enhancing technologies on targeting accuracy in two real-world settings. (A) Anti-poverty program in Togo. In simulations of a nationwide humanitarian program, differential privacy ($\epsilon \approx 4$) and k-anonymity ($k=2$) would increase exclusion errors by $115K$ and $240K$, relative to the non-private status quo. Our approach (targeted differential privacy) increases exclusion errors by $2K$ ($B = 0.25$). (B) Micro-lending platform in Nigeria. Simulating the accuracy of credit scoring algorithms used to extend loans to individuals without a formal financial history, differential privacy ($\epsilon \approx 3$)  and $k$-anonymity ($k=2$) would reduce the profits of the program by 430\% and 476\%, respectively. Targeted differential privacy would reduce profits by 12\% ($B = 0.1$).}
\label{fig:pets_accuracy}
\end{figure}

\begin{figure}[!p]
\centering
\includegraphics[width=\textwidth]{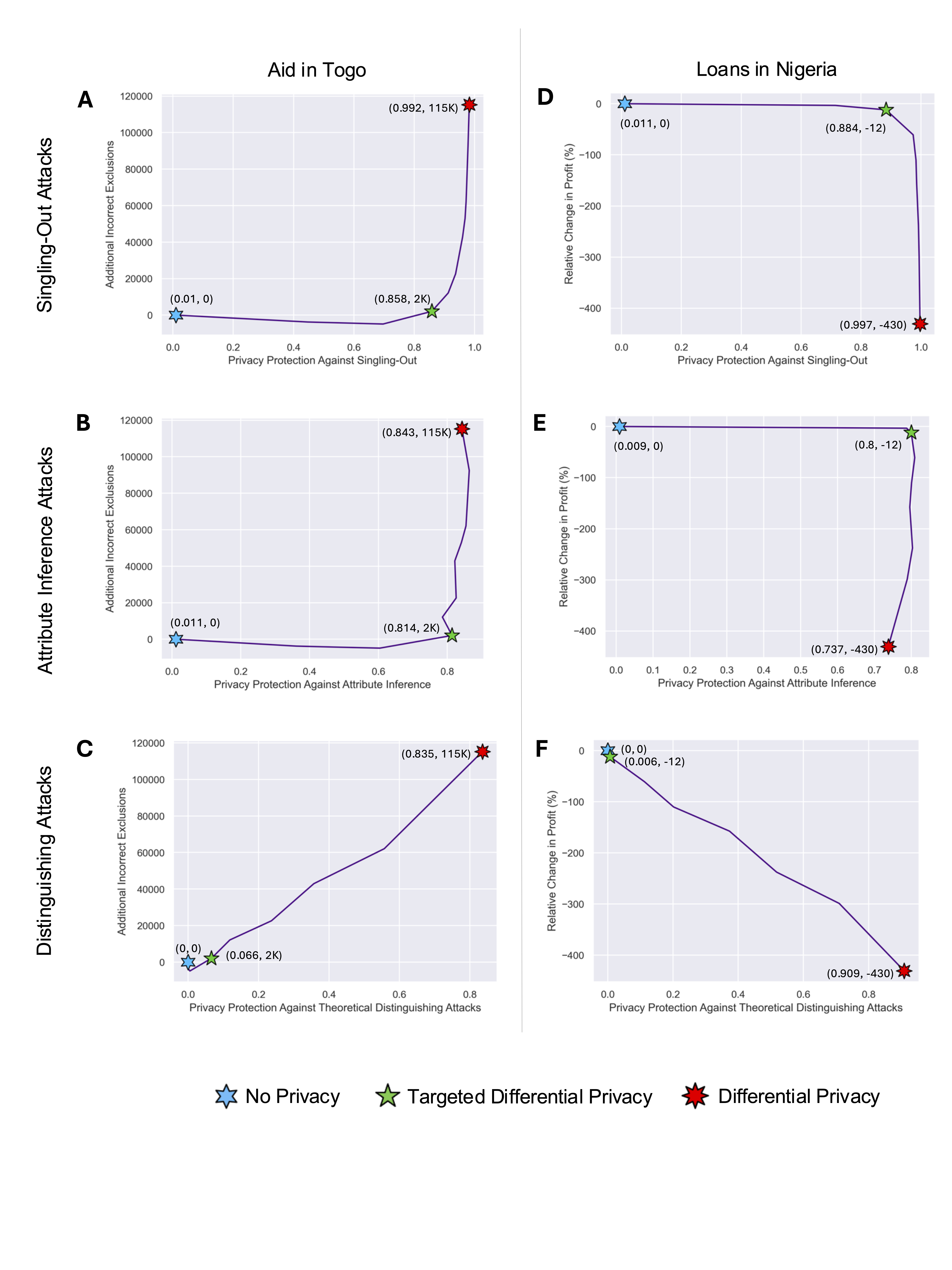}
\vspace*{-30mm}
\caption{Empirical tradeoffs between privacy and program effectiveness for a humanitarian program in Togo (left panels) and a digital lending platform in Nigeria (right panels). Privacy protections are shown for singling-out attacks (top row), attribute inference attacks (middle row), and distinguishing attacks (bottom row).  For all panels, blue stars represents the non-private status quo, and red stars represents  classic $(\epsilon, \delta)$-differential privacy. Green stars corresponds to targeted differential privacy, with $B = 0.25$ and $\epsilon\approx 4$ for the Togolese program and $B = 0.1$ and $\epsilon \approx 3$ for the Nigerian application.}
\label{fig:tradeoffs}
\end{figure}

\newpage

\section*{Supplemental Information}

\section{Supplementary Note: Targeted Differential Privacy}
\label{maths}
\subsection{Formalizing the targeting problem and setting}
\label{app:math_formalize}

Our framework of targeted differential privacy is motivated by an interaction between a data holder (such as a mobile network operator) and a downstream program official who wishes to make use of the data. The program official seeks to target a subset of $n$ individuals for some benefit, based on a \textit{targeting variable} $y$. This targeting variable could be an individual's or household's income \cite{deaton1997analysis, bradbury2004targeting}, their daily or household consumption \cite{blumenstock2016fighting, aiken2022machine}, or their credit score \cite{bjorkegren2020behavior}. However, the program official faces an information restriction: they only possess target variables for a non-empty proper subset $I \subset [n] = \{1,...,n\}$ (these individuals are \textit{in the training sample}). Call these targets $y_I \in \R^{|I|}$. Let $O = [n] \cap I^{c}$ denote the subset of individuals \textit{out of the training sample}.

In certain humanitarian situations, it is infeasible to directly solicit targeting information on $O$  \cite{blumenstock2016fighting, aiken2022machine, aiken2023program}. To overcome these informational constraints, program officials have partnered with data holders (such as a mobile network operator) to solicit data on individuals in $[n]$ to predict these targeting variables using techniques from machine learning \cite{blumenstock2016fighting, aiken2022machine, aiken2023program}. We represent such data as an $n \times d$ matrix $X$, where each row corresponds to $d$ features about individual $i \in [n]$. For notational purposes, we let $X_i$ denote the $i^{th}$ row of $X$, and $X_{(j)}$ denote the $j^{th}$ column of $X$. Following this convention, we let $X_I$ denote the row-wise submatrix of $X$ that contains data on the individuals in $I$. For ease of exposition, we will use the terms matrices, datasets, and databases interchangeably. 

Equipped with $X_I$ with $y_I$, the program official constructs a machine learning model $M$, which can then be used to generate \textit{predicted targeting variables} $\hat{y}$. In practice, a value of $\hat{y}$ is generated for every individual in $[n]$, even if their actual targeting variable is known to the program official. The rationale is that this provides a level of fairness in the sense that the eligibility of all individuals is determined by the same targeting process.

For our study, we consider the setting where the data holder -- cognizant of privacy concerns -- does not share a dataset $X$ with the program official.\footnote{In the setting we consider, the program official does not share $y_I$ with the data holder. In situations where the data holder can share $y_I$ as part of a secure computing protocol, it is possible to utilize multiple privacy-enhancing technologies to generate $\hat{y}$. For example, one possible design solution would be to create a differentially private machine learning model $M$ using secure multiparty computation, and then predict each individual's targeting variable using $X$.} Instead, the data holder provides the program official with a \textit{provably private} dataset $X_{priv}$, which will be used in place of $X$ in the learning and prediction process described above. The dataset $X_{priv}$, which we will refer to as a \textit{privatized dataset}, is the output of some privacy-enhancing technology $A$. 

Our study seeks to understand the properties $A$ must have to ensure $X_{priv}$ (1) yields accurate predictions to enable effective programmatic outcomes, while (2) providing strong privacy guarantees, as well as characterize the limits and tradeoffs induced by $A$. 

\subsection{Targeted differential privacy}
\label{app:math_tdp}


To describe targeted differential privacy, we will utilize the following notation and terminology. Let $||\cdot||_2$ denote the $L_2$-norm, $\L_2^{d}$ denote the the set of $d$-dimensional real vectors with $L_2$ at most 1, and $\L_2^{n \times d}$ denote the set of of real $n \times d$ matrices where the $L_2$ norm of each row is at most 1. 

We say two databases $X, X' \in \L_2^{n\times d}$ are \textit{classic neighbors} if they agree on exactly $n-1$ rows. Classic differential privacy ensures that an algorithm's behavior is statistically indistinguishable for all pairs of classically neighboring datasets \cite{dwork2006our}.

\begin{defn}
\label{defn: cdp}
    A randomized algorithm $A$ from $\L_2^{n\times d}$ to $\mathbb{O}$ satisfies $(\epsilon,\delta)$-\textit{classic differential privacy} if for all classic neighbors $X,X' \in \L_2^{n\times d}$ and for all measurable sets $E \in \mathbb{O}$, 
    $$\P(A(X) \in E) \le e^{\epsilon}\P(A(X') \in E) + \delta$$ 
For brevity, at times we will omit the term ``classic'' and simply refer to this condition as \textit{differential privacy} (DP).
\end{defn}

For machine learning applications, the specific values in the row of a dataset can affect the predictions made. As such, it will be helpful to consider a more granular notion of neighboring datasets that quantities the extent to which two rows disagree. 

We say  $X, X' \in \L_2^{n\times d}$ are \textit{$B$-neighbors} if they are classic neighbors, and for the single row $i$ where they disagree, $||X_i - X'_i||_2 \le B$. When $B = 2$, $B$-neighbors coincides with classic neighbors. As we will see, $B$ serves as an analytical device, enabling a finer analysis of privacy loss compared to classic differential privacy.

We define the cumulative distance between databases $X$ and $X'$, both of size $n$, as the sum of their row-wise distances. That is, if $X$ and $X'$ have $n$ rows, their cumulative distance is $d(X,X') = \sum_{i \in [n]} ||X_i - X'_i||_2$. By definition, if $X$ and $X'$ are $B$-neighbors, then $d(X,X') \le B$. 

Targeted differential privacy is a variant of classic differential privacy that requires an algorithm's behavior satisfy statistical indistinguishability for all pairs of $B$-neighboring datasets.

\begin{defn}
\label{defn: tdp}
    A randomized algorithm $A$ from $\L_2^{n\times d}$ to $\mathbb{O}$ satisfies $(B, \epsilon,\delta)$-\textit{targeted differential privacy}\footnote{The definition of TDP has appeared in prior privacy studies in different contexts. For additional information, see \textit{Related Contextual Adaptations} in the \textit{Methods} section of the main manuscript.} (TDP) if for all $B$-neighbors $X,X' \in \L_2^{n\times d}$ and for all measurable sets $E \in \mathbb{O}$, 
    $$\P(A(X) \in E) \le e^{\epsilon}\P(A(X') \in E) + \delta$$
By definition, when $B = 2$ we recover classic differential privacy. 
\end{defn}

The proofs of post-processing, sequential composition, and parallel composition are nearly identical to those presented in Dwork and Roth \cite{dwork2014algorithmic} and McSherry \cite{mcsherry2009privacy}, so we omit them.

\begin{lem}
\label{lem:postprocess}
    [Post-Processing] Suppose $A$ from $\L_2^{n\times d}$ to $\O$ satisfies $(B,\epsilon, \delta)$-TDP. Then for any (potentially randomized) function $F$ from $\O$ to some space $\O'$, $F \circ A$ also satisfies $(B,\epsilon, \delta)$-TDP.
\end{lem}

\begin{lem}
\label{lem:seqcom}
    [Sequential Composition] Suppose randomized algorithms $A_j$ mapping $\L_2^{n\times d}$ to $\O_j$ satisfies $(B_j,\epsilon_j, \delta_j)$-TDP for $j \in [m]$. Then the sequential composition of these mechanisms, $(A_1,...,A_m)$ from $\L_2^{n\times d}$ to $\prod_{j \in [m]} O_j$ satisfies $(\min_{j \in [m]}B_j, \sum_{j \in [m]} \epsilon_j, \sum_{j \in [m]} \delta_j)-TDP$
\end{lem}

\begin{lem}
\label{lem:parcom}
    [Parallel Composition] Suppose $A_j$ from $\L_2^{n\times d}$ to $\O_j$ satisfies $(B_j,\epsilon_j, \delta_j)$-TDP. If these algorithms are computed on disjoint sets of data, the composite mechanism $(A_1,...,A_m)$ from $\L_2^{n\times d}$ to $\prod_{j \in [m]} \O_j$ satisfies $(\min_{j \in [m]}B_j, \max_{j \in [m]} \epsilon_j, \max_{j \in [m]} \delta_j)$-TDP.
\end{lem}

Privacy protection under TDP decays predictably for datasets $x$ and $y$ that are classic neighbors.

\begin{lem}
\label{lem:switch}
    [Switch Lemma] Suppose a randomized algorithm $A$ from  $\L_2^{n\times d}$ to $\mathbbm{O}$ satisfies $(B,\epsilon, \delta)$-TDP with $B, \epsilon >0$. Then for all classic neighboring datasets $X,\hat{X} \in \L_2^{n\times d}$, and for all measurable sets $E \in \O$,

    $$
    \P(A(X) \in E) \le e^{s\epsilon}\P(A(\hat{X}) \in E) + \frac{e^{s\epsilon}-1}{e^{\epsilon}-1} \delta
    $$
    where $s = \ceil{d(X,\hat{X})B^{-1}}$
\end{lem}

\begin{proof}
    Consider an arbitrary measurable set $E \subseteq \O$. If $d(X,\hat{X}) \le B$, the claim follows immediately. So consider $d(X,\hat{X}) > B$. Since $X$ and $\hat{X}$ are classic neighbors, they differ in one element: denote this as the the $t^{th}$ element. Then $||X_t - \hat{X}_t||_2 = d(X,\hat{X})$. 

    Next, we will relate $X$ and $\hat{X}$ to an sequence of $B$-neighboring datasets; to do so, we will first construct a preliminary sequence of datapoints starting at $X_t$ and ending at $\hat{X}_t$ whose adjacent values in the sequence are at most distance $B$ apart. To motivate the construction used in the upcoming paragraph, consider the line segment in $\L_2^d$ whose endpoints are $X_t$ and $\hat{X}_t$. This line segment has length $d(X,\hat{X})$. We will break this line segment into $m+1$ segments that are each at most length $B$. Since $\L_2^d$ is convex, these breakpoints must reside in $\L_2^d$, and will be used these as values in our preliminary sequence. 
    
    This is done mathematically by letting $m$ be the smallest non-negative integer such that $d(X,\hat{X}) - mB \le B$. Then $m$ is the smallest non-negative integer such that $m \ge d(X,\hat{X})B^{-1}-1$. So $m = \ceil{d(X,\hat{X})B^{-1}-1}$, which implies $m+1 = \ceil{d(X,\hat{X})B^{-1}}$. Consider a sequence $X_t^{(0)},X_t^{(1)},...,X_t^{(m+1)} \in \L_2^{d}$, defined as follows: $X_t^{(0)} = X_t$, $X_t^{(m+1)} = \hat{X}_t$, and the remaining $X_t^{(i)}$ are any values such that $||X_t^{(i)} - X_t^{(i+1)}||_2 \le B$ for $i \in \{0,...,m\}$, provided they are consistent with the endpoints $X_t^{(0)}$ and $X_t^{(m+1)}$.
    
    Now, define $X^{(i)}$ as the dataset $X$ with $X_t$ replaced with $X_t^{(i)}$. This constructs the sequence of datasets $X^{(0)}, X^{(1)}, ..., X^{(m+1)}$. By construction, all of these datasets are in $\L_2^{n \times d}$; also, $X = X^{(0)}$ and $\hat{X} = X^{(m+1)}$. 
    
    Write $p_i = \P(A(X^{(i)}) \in E)$. By construction, every adjacent pair of datasets $X^{(i)}$ and $X^{(i+1)}$ in our sequence are $B$-neighbors, so
    \begin{align*}
        p_0 &\le e^{\epsilon}p_1 + \delta \\
        &\le e^{\epsilon}(e^{\epsilon}p_2 + \delta) + \delta \\
        &= e^{2\epsilon}p_2 + (e^{\epsilon}+1)\delta \\
        & ... \\
        & \le e^{(m+1)\epsilon}p_{m+1} + (e^{m\epsilon}+...+e^{\epsilon}+1)\delta \\
        &= e^{(m+1)\epsilon}p_{m+1} + \frac{e^{(m+1)\epsilon}-1}{e^{\epsilon}-1}\delta
    \end{align*}
    Since $m+1 = s$, $p_0 = \P(A(X) \in E)$, and $p_{m+1} = \P(A(\hat{X}) \in E)$, the claim is proven.
\end{proof}

The switch lemma provides a bridge between targeted and classical differential privacy. Namely, if $A$ satisfies $(B,\epsilon, \delta)$-TDP then $A$ also satisfies $(\epsilon',\delta')$-DP, where $\epsilon' = s\epsilon$, $\delta' = \min\left\{1,\frac{e^{s\epsilon}-1}{e^{\epsilon}-1}\delta\right\}$, and $s = \ceil{2B^{-1}}$. However, the converse is not necessarily true: TDP requires the $(\epsilon,\delta)$ probability bound holds for every diameter $B$ ball; however, $(\epsilon',\delta')$-DP does not necessarily imply the $(\epsilon,\delta)$ probability bound holds in every diameter $B$ ball. As such, there is a key semantic distinction between classic and targeted differential privacy: while TDP guarantees an adversary cannot discern (up to the factors $\epsilon$ and $\delta$) between two individuals' datapoints that are $B$-similar, under $(\epsilon',\delta')$-DP an adversary may be able to do so.

\subsection{Formal analysis of the necessary conditions for accurate targeting}
\label{app:math_necessary}

In this section, we formalize the relationship between the privacy parameters $(B,\epsilon, \delta)$ and the the goal of accurate targeting. We consider the binary setting, where an individual is classified either as eligible for a benefit ($y=1$) or not ($y = 0$). Denote the set $\{0,1\}$ as $\Z_2$.

Let $T:\L_2^{n \times d} \rightarrow \Z_2$ denote the targeting process for an arbitrary individual. In the setting we consider in the main manuscript, $T(X)$ is the result of running a learning algorithm on $(X_I,y_I)$, outputting the machine learning model, and then using this model to determine an individual's eligibility status. When $A(X)$ outputs the privatized dataset $X_{priv}$ that is used in place of $X$ for targeting, the resulting eligibility classification is given by $T(A(X))$. 

We restrict our analysis to targeting processes $T$ that are \textit{minimally responsive}: these are functions for which there exist classic neighboring datasets $X$ and $X'$ such that $T(X) \ne T(X')$. Note that if $T$ is not minimally responsive, then $T$ outputs the same  predictions regardless of an individual's data. Such targeting processes are not useful in practice, and hence we omit them from our analysis.

To enable accurate targeting, we need the mechanism $T \circ A$ to behave similarly to $T$ with high probability. Formally, we say $A$ is \textit{$\gamma$-accurate} for $T$ if for all $X \in \L_2^{n \times d}$, we have $\P(T(A(X)) = T(X)) \ge \gamma$.


We now characterize the necessary conditions for accurate targeting under minimally responsive targeting processes.

\begin{thm}
\label{thm:packing_classification}
    Let $A$ be a randomized algorithm from $\L_2^{n \times d}$ to $\L_2^{n \times d}$ satisfying $(B,\epsilon,\delta)$-TDP with $B, \epsilon > 0$. Suppose that $T: \L_2^{n \times d} \rightarrow \Z_2$ is a minimally responsive deterministic function. If $A$ is $\gamma$-accurate for $T$ with $\gamma \in \left[\frac{1}{2},1\right)$, then 
    $$ 
    \ceil{2B^{-1}} \ge \ceil{\epsilon^{-1}\ln(Q)}
    $$
    where 
    $$
    Q = \frac{\delta + \gamma(e^{\epsilon}-1)}{\delta + (1-\gamma)(e^{\epsilon}-1)}
    $$ 
    In particular, when $2B^{-1} \in \Z$ and $\gamma > \frac{1}{2}$, the inequality simplifies to 
    $$ 
    B \le \frac{2}{\ceil{\epsilon^{-1}\ln(Q)}}
    $$
\end{thm}

\begin{proof}    
    By the minimal responsiveness of $T$, there exists classic neighboring datasets $X$ and $X'$ in $\L_2^{n\times d}$ such that $T(X) \ne T(X')$. Since $A$ satisfies $(B,\epsilon,\delta)$-TDP, so too does $T\circ A$ by the post-processing lemma (Lemma \ref{lem:postprocess}). By the switch lemma (Lemma \ref{lem:switch}), we have 
    $$
    \P(T(A(X)) \in E) \le e^{s\epsilon}\P(T(A(X')) \in E) +\frac{e^{s\epsilon}-1}{e^{\epsilon}-1} \delta 
    $$
    for all sets $E \subseteq \Z_2$, where $s = \ceil{d(X,X')B^{-1}}$. In particular, consider the set $E = \{T(X)\}$. 
    
    By the accuracy condition, $\P(T(A(X)) \in E) \ge \gamma$. Since $E^c = \{T(X')\}$, the accuracy condition also implies $\P(T(A(X')) \in E^c) \ge \gamma $, so  $\P(T(A(X')) \in E) = 1 - \P(T(A(X')) \in E^c) \le 1- \gamma$. Combining this with the above inequality, we have
    $$
    \gamma \le e^{s\epsilon}(1-\gamma)+\frac{e^{s\epsilon}-1}{e^{\epsilon}-1} \delta 
    $$

        Rearranging terms and solving for $s$ yields $s \ge \epsilon^{-1}\ln(Q)$. Note that the quantity on the right-hand side of the inequality is well-defined, as $\gamma \ge \frac{1}{2}$ implies $Q \ge 1$. Next, we deduce that $s \ge \ceil{\epsilon^{-1}\ln(Q)}$. There are two cases to consider.

    \begin{itemize}
        \item[$\star$] \textit{Case 1}: If $\epsilon^{-1}\ln(Q) \in \Z$, then $\epsilon^{-1}\ln(Q) = \ceil{\epsilon^{-1}\ln(Q)}$, so $s \ge \ceil{\epsilon^{-1}\ln(Q)}$ holds. 
        \item[$\star$] \textit{Case 2}: Otherwise, if  $\epsilon^{-1}\ln(Q) \in \R \cap \Z^c$, then $s \ne \epsilon^{-1}\ln(Q)$ since the left-hand side is an integer and the right hand side is not. Hence we must have $s > \epsilon^{-1}\ln(Q)$, implying $s \ge \ceil{\epsilon^{-1}\ln(Q)}$.
    \end{itemize} 
     Therefore, in either case we have $s \ge \ceil{\epsilon^{-1}\ln(Q)}$. Since $s = \ceil{d(X,X')B^{-1}}$ and $d(X, X') \le 2$, we have  $s \le \ceil{2B^{-1}}$ because the ceiling function is monotonically non-decreasing. Therefore $\ceil{2B^{-1}} \ge \ceil{\epsilon^{-1}\ln(Q)}$ as claimed.
     
     Furthermore, when  $2B^{-1} \in \Z$, the inequality simplifies to $2B^{-1} \ge \ceil{\epsilon^{-1}\ln(Q)}$. Also, $\gamma > \frac{1}{2}$ implies $\ln(Q) > 0$, so $\ceil{\epsilon^{-1}\ln(Q)} \ge 1$. Solving for $B$ yields 
     $$
     B \le \frac{2}{\ceil{\epsilon^{-1}\ln(Q)}}
     $$  
\end{proof}

\subsection{Private projection algorithm details}
\label{app:math_algorithm_details}

Next, we describe an algorithm that satisfies $(B, \epsilon, \delta)$-TDP. Algorithm \ref{algo:edprp} is motivated by the Johnson-Lindenstrauss lemma \cite{chen2015johnson, nabil2017random}, and its differentially private variants \cite{kenthapadi2012privacy, blocki2012johnson, gondara2020differentially}. Algorithm \ref{algo:edprp} is based on Gondara and Wong's differentially private randomized projection method \cite{gondara2020differentially}, with three notable alterations (two of which adapt to the projection phase, and the remaining adapts the covariance projection stage). The first alteration is in Step 1 of the algorithm. Gondara and Wong sample the elements of $R$ from a Gaussian with mean $0$ and variance $k^{-1}$. Our algorithm instead samples values from $\{-1,0,1\}$ uniformly at random. This improves the practical efficiency of our algorithm by utilizing integer computations over floating point computations \cite{nabil2017random, achlioptas2001database, bingham2001random}. The second alteration occurs in Step 2 of the algorithm where the projection values are inversely scaled by $k$, which reduces the standard deviation of the noise used as $k$ increases in Step 3 to achieve TDP. The third alteration occurs in Step 4 of the algorithm. In their algorithm, Gondara and Wong construct a noisy covariance matrix $C_{priv}$ by summing together the covariance matrix $X^{T}X$ and a matrix whose $G$ whose values are drawn from a Gaussian distribution. Since the matrix $G$ may not be symmetric, their resulting $C_{priv}$ may not be symmetric. For this reason we follow the approach in Dwork et al. \cite{dwork2014analyze} and construct a perturbation matrix $G$ whose upper triangle values are sampled from a Gaussian distribution, and whose lower triangular values are copied from the upper triangle. By construction, $G$ is now symmetric. Since the sum of symmetric matrices is symmetric, this design alteration has the benefit of constructing a symmetric $C_{priv}$.



The remainder of this section is devoted to showing that Algorithm \ref{algo:edprp} satisfies $(B,\epsilon,\delta)$-TDP. We begin with the following lemma, which examines the sensitivity of the matrix multiplication by $R$ in Step 2, as measured by the Frobenius norm $||\cdot||_F$.

\begin{lem}
\label{lem:global_sens}
    Suppose $X,X'\in \mathbb{L}_2^{n \times d}$ are $B$-neighbors. For any matrix $R \in \R^{d \times k}$,
    $$
    ||k^{-1}XR - k^{-1}X'R||_{F}^{2} \le \frac{B^2}{k^2}\sum_{j=1}^k \sum_{p=1}^d R_{p,j}^2
    $$
\end{lem}

\begin{proof}
    Let $R \in \R^{d \times k}$ be given. Since $X$ and $X'$ are $B$-neighbors, there exists a unique row, say $i \in [n]$, such that $X_i - X'_i \ne 0$. Consider the $(p,j)$ entry of $XR - X'R$. This is given by $\mathbbm{1}(p = i)(X_p - X'_p)R_{(j)}$. So then
    \begin{align*}
        ||XR - X'R||_{F}^{2} & =  \sum_{j=1}^k |(X_i - X'_i)R_{(j)}|^2 \\
        & \le \sum_{j=1}^k ||X_i - X'_i||_2^2 ||R_{(j)}||_2^2 \\
        & = ||X_i - X'_i||_2^2 \sum_{j=1}^k  ||R_{(j)}||_2^2 \\
        & \le B^2 \sum_{j=1}^k  ||R_{(j)}||_2^2 \\
        & = B^2 \sum_{j=1}^k \sum_{p=1}^d R_{p,j}^2
    \end{align*}
    where the first inequality follows by the Cauchy-Schwartz inequality and the second follows by the definition of $B$-neighbors. Hence, $||XR - X'R||_{F}^{2} \le B^2 \sum_{j=1}^k \sum_{p=1}^d R_{p,j}^2$. Multiplying both sides of the inequality by $k^{-2}$ produces the claim.
\end{proof}


In particular, when $R \in \{-1, 0 ,1 \}^{d \times k}$, the inequality in Lemma \ref{lem:global_sens} reduces to $||k^{-1}XR - k^{-1}X'R||_{F} \le B\sqrt{d/k}$ in the worst-case. Rather than scale the noise in Step 3 of our algorithm using this worst-case bound,  we can reduce the amount of noise introduced by considering a probabilistic version of the Frobenius norm based on the randomness induced by $R$. That is, we can instead guarantee $||k^{-1}XR - k^{-1}X'R||_{F} \le \alpha$ except with some failure probability $\beta$, which we incorporate in our privacy analysis.


\begin{prop}
\label{prop:probablistic_global_sens}
    Suppose each element of $R \in \R^{d \times k}$ is sampled i.i.d from a categorical distribution over $\{-1,0,1\}$ as follows: for $p_0 \in (0,1)$,
    \begin{center}
        \begin{tabular}{ c|c|c|c } 
        z & $-1$ & $0$ & $1$ \\ 
        \hline
        $\P(z)$ & $\frac{1-p_0}{2}$ & $p_0$ & $\frac{1-p_0}{2}$ \\ 
        \end{tabular}
    \end{center}
    Then, $\P(||k^{-1}XR-k^{-1}X'R||_F \ge \alpha) \le \beta$ whenever 
    $$
    \alpha \ge \frac{B}{\sqrt{k}} \sqrt{d\ln((1-p_0)(e-1) + 1) -k^{-1}\ln(\beta)}
    $$
\end{prop}

\begin{proof}
    Suppose the $(p,j)$ entry of $R$ is chosen according to the distribution described in the proposition. Then $R_{p,j}^2$ is $0$ with probability $p_0$ and $1$ with probability $1-p_0$, so $R_{p,j}^2 \sim \text{Bernoulli}(1-p_0)$. Let $Y = \sum_{j=1}^k \sum_{p=1}^d R_{p,j}^2 $. Then $Y \sim \text{Binomial}(dk, 1-p_0)$.

    Next, observe that $\P(||k^{-1}XR-k^{-1}X'R||_F \ge \alpha) = \P(||k^{-1}XR-k^{-1}X'R||_F^2 \ge \alpha^2)$. By Lemma \ref{lem:global_sens}, if $||k^{-1}XR-k^{-1}X'R||_F^2 \ge \alpha^2$ then $B^2k^{-2}\sum_{j=1}^k \sum_{p=1}^d R_{p,j}^2 \ge \alpha^2$. Hence, 
    
    $$
    \P(||k^{-1}XR-k^{-1}X'R||_F^2 \ge \alpha^2) \le \P\Bigg(B^2k^{-2} \sum_{j=1}^k \sum_{p=1}^d R_{p,j}^2 \ge \alpha^2\Bigg) = \P(B^2 k^{-2}Y \ge \alpha^2)
    $$
    
    Therefore,    
    \begin{align*}
        \P(||k^{-1}XR-k^{-1}X'R||_F \ge \alpha) & \le \P(B^2 k^{-2}Y \ge \alpha^2) \\
        & = \P\Bigg(Y \ge \Big(\frac{\alpha k}{B}\Big)^2\Bigg) \\
        & \le \frac{\E[\exp{(Y)}]}{\exp{\Big(\Big(\frac{\alpha k}{B}\Big)^2\Big)}} \\ 
        & = \frac{((1-p_0)(e-1) + 1)^{dk}}{\exp{\Big(\Big(\frac{\alpha k}{B}\Big)^2\Big)}}
    \end{align*}
    where the second inequality follows by the Markov bound, and the last equality follows from the moment generating function of the binomial distribution. Plugging in the bound for $\alpha$ in the proposition yields the claim.

\end{proof}

Now, set $p_0 = \frac{1}{3}$ and $\beta = \delta/2$ in Proposition \ref{prop:probablistic_global_sens}. Then,

$$
\alpha \ge  \frac{B}{\sqrt{k}}\sqrt{d\ln((2/3)(e-1) + 1) -k^{-1}\ln(\delta/2)}
$$

With this probabilistic sensitivity bound, we use the following lemma from Kenthapadi et. al. \cite{kenthapadi2012privacy} to show that Steps 1-3 of Algorithm \ref{algo:edprp} satisfy $(B, \epsilon_1,\delta_1)$-TDP.

\begin{lem}
\label{lem:kenthapadi}
    [Lemma 1 from \cite{kenthapadi2012privacy}; Lemma 3 from \cite{gondara2020differentially}] The mechanism given by $M(X) = f(X) + G$ satisfies $(B, \epsilon, \delta)$-TDP if $\delta < \frac{1}{2}$ and $G \in \R^{n \times k}$ with each $G_{p,j} \sim_{i.i.d} N(0, \sigma^2)$, where
    $$
    \sigma^2 = 2\Delta_2(f)^2\frac{\ln(1/(2\delta)) + \epsilon}{\epsilon^2} 
    $$
    and $\Delta_2(f)$ is the global sensitivity of $f$.
\end{lem}

\begin{prop}
\label{prop:projection}
    Set  $p_0 = \frac{1}{3}$ and $\beta = \delta/2$ in the projection distribution of Proposition \ref{prop:probablistic_global_sens}. Fix a matrix $R$ generated according to this distribution. Then the mechanism given by $M(X) = XR + G$ satisfies $(B, \epsilon, \delta)$-TDP if $\delta < \frac{1}{2}$ and $G \in \R^{n \times k}$ with each $G_{p,j} \sim_{i.i.d} N(0, \sigma^2)$, where
    $$
    \sigma = \frac{B}{\sqrt{k}}\sqrt{d\ln((2/3)(e-1) + 1) -k^{-1}\ln(\delta/2)}
    \frac{\sqrt{2(\ln(1/\delta) + \epsilon)}}{\epsilon} 
    $$ 
\end{prop}

\begin{proof}
    By the Proposition \ref{prop:probablistic_global_sens}, with probability $1-\beta$, $||k^{-1}XR-k^{-1}X'R||_F \le \alpha$. Set $\beta = \frac{\delta}{2}$. Now, take $f(X) = k^{-1}XR$. Then $\Delta_2(f) = ||k^{-1}XR-k^{-1}X'R||_F \le \alpha$ with probability $1-\frac{\delta}{2}$. Plugging in the value of $\alpha$ in Lemma \ref{lem:kenthapadi} with $\delta/2$ yields the claim.
\end{proof}

Next, we show that Step 4 of Algorithm \ref{algo:edprp} satisfies $(B, \epsilon_2,\delta_2)$-TDP.

\begin{lem}
    \label{lem:gauss_mech_tdp} 
    Let $C$ be the covariance matrix of $X \in \L_2^{n \times d}$. Then the mechanism $C + G$ satisfies $(B,\epsilon, \delta)$-TDP where $G_{i,j} \sim_{i.i.d} N(0, 2g_C^2\ln(1.25/\delta)/\epsilon^2)$ for all $i \ge j$, and $G_{i,j} =  G_{j,i}$ for all $i < j$, where  $g_C \le 2B$ is the worst-case Frobenius norm for the difference of covariance matrices that are $B$-neighbors, provided $\epsilon \in (0,1)$.
\end{lem}

\begin{proof}
    Suppose $X$ and $Z$ are $B$-neighbors. Then there exists a unique row $i$ such that $||X_i - Z_i||_2 \le B$, with all other rows being identical between $X$ and $Z$. For notational simplicity, denote $X_i = x$ and $Z_i = z$. The  sensitivity of the covariance computation is $||X^{T}X - Z^{T}Z||_F = ||x^{T}x - z^{T}z||_F$. Then
    \begin{align*}
        ||X^{T}X - Z^{T}Z||_F & = ||x^{T}x - z^{T}z||_F \\
        & = ||(x^{T}x -x^{T}z) + (x^{T}z - z^{T}z)||_F \\
        & = ||x^T(x-z) + (x-z)^Tz||_F \\
        & \le ||x^T(x-z)||_F + ||(x-z)^Tz||_F \\
        & = ||x||_2 ||x-z||_2 + ||x-z||_2||z||_2 \\
        & \le 2B
    \end{align*} 
    Hence $g_C \le 2B$. Since $\epsilon \in (0,1)$, we can use this value as the sensitivity in the Gaussian perturbation mechanism \cite{dwork2006our, dwork2014algorithmic, balle2018privacy}. The claim now follows by the same line of reasoning as Theorem 2 of Dwork et al. \cite{dwork2014analyze}.
\end{proof}

Taken together, the results of this section all culminate to the following theorem. 

\begin{thm}
\label{thm:alg_satisfies_tdp}
    Algorithm \ref{algo:edprp} satisfies $(B,\epsilon_1 + \epsilon_2,\delta_1+\delta_2)$-TDP.
\end{thm}

\begin{proof}
    By the Proposition \ref{prop:projection}, Steps 1-3 of Algorithm \ref{algo:edprp} satisfy $(B,\epsilon_1,\delta_1)$-TDP. By Lemma \ref{lem:gauss_mech_tdp}, Step 4 satisfies $(B,\epsilon_2,\delta_2)$-TDP. By the post-processing lemma (Lemma \ref{lem:postprocess}), Step 5 does not increase the privacy loss from Step 4. Step 6 combines the computed values, so by the sequential composition lemma (Lemma \ref{lem:seqcom}), our algorithm satisfies $(B,\epsilon_1 + \epsilon_2,\delta_1+\delta_2)$-TDP. 
\end{proof}

\section{Supplementary Note: Additional Information on Privacy Measures}
\label{protection}
\subsection{Relative protection score using the holdout-approach for attribute inference}
\label{aia}


In this section, we provide additional technical details on our quantification of attribute inference protection. We follow the approach presented in Giomi et al. \cite{giomi2022unified}. We consider an adversary who has access to the privatized data $X_{priv}$, the privatization algorithm $A$, and an auxiliary information source: a column-wise submatrix $S$ of $1 \le h < d$ columns of $X$. That is, the adversary has access to $h$ columns of the original data on all $n$ individuals. Using these three objects, the adversary's goal is to infer the missing $d-h$ columns of $X$.

Our measurement of attribute inference protection is adopted from Giomi et al. \cite{giomi2022unified}, which uses a holdout-approach similar to machine learning applications. For completeness and self-contained exposition, we describe the measurement procedure (and provide additional details) below.

For a column $j$ that an adversary seeks to infer in a database $X$, let $p(Z;S)_j$ represent $Z$'s \textit{protection score for column $j$} when $S$ is known. As noted in \textit{Methods} in the main text, this the proportion of individuals whose values the adversary correctly infers in column $j$ (up to a tolerance of 5\% error) using a privatized database $Z$ and auxiliary info $S$. 

Given the original data $X$, randomly select 500 individuals to remove from $X$. Denote $X$ without these 500 individuals as $W$ (working-set) and the dataset with these 500 individuals as $H$ (holdout-set). Run the privatization algorithm $A$ on $W$ to construct $W_{priv}$. Run the attribute inference program procedure on both $W_{priv}$ and $H$ to generate $p(W_{priv};S)_j$ and $p(H;S)_j$.

For a column $j$, the protection score $p(W_{priv};S)_j$ in isolation is difficult to interpret. For instance, if $p(W_{priv};S)_j$ is small, then this could be attributed to multiple culprits. In one case, $A$ preserved population-level information in $W$ well without retaining ``too much information'' about any single individual; since the population level information is preserved, individual values could nonetheless be inferred. On the other hand, $A$ could have retained ``too much information'' about an individual, enabling the inference of their values. In the first example, the success rate misidentifies the ``statistical utility'' of our privatized dataset with inference risk, whereas in the second example the success rate correctly identifies inference risk.

In order to delineate between these two cases, we compare $p(W_{priv};S)_j$ with $p(H;S)_j$. Since $W_{priv}$ was created using $W$ and not $H$, if $p(W_{priv};S)_j$ and $p(H;S)_j$ are close, then this can be ascribed to the algorithm $A$ capturing statistical information from the whole working population $W$ and not from a particular individual. Alternatively, if $p(W_{priv};S)_j << p(H;S)_j$, then it likely has learned ``too much'' about particular individuals. Thus, we define the \textit{relative protection score of column $j$ under $S$} as\footnote{Giomi et al. define this an equivalent quantity in terms of a \textit{risk score}, which is $1-r_{j,S}$.} 
$$
r_{j,S} = \frac{p(W_{priv};S)_j}{p(H;S)_j}
$$

\subsection{Computing the distinguishing protection score for our algorithm}
\label{dist}

Our privacy measure of distinguishing protection is rooted in the \textit{privacy loss random variable}, a  unitless quantity that measures the differences in an algorithm's behavior between classic neighbors $X,X' \in \L_2^{n \times d}$ \cite{dwork2014algorithmic}. The privacy loss random variable quantifies the adversary's advantage during the distinguishing game when a randomized algorithm $A$ outputs $X_{priv}$ (i.e., the ability to discern whether $X$ or $X'$ was more likely to produce $X_{priv}$ under $A$) \cite{desfontainesblog20200306}. 

For Gaussian perturbation mechanisms with variance $\sigma^2$, the privacy loss random variable follows a Gaussian distribution with mean $||f(X) - f(X')||_F^2 / (2\sigma^2)$ \cite{balle2018privacy}. Therefore, the worst-case expected privacy loss is given by $\text{sup}||f(X) - f(X')||_F^2 / (2\sigma^2)$ where the supremum is taken over $X$ and $X'$ classic neighbors.

Algorithm \ref{algo:edprp} is comprised of two Gaussian perturbation mechanisms in Steps 3 and 4. Let $E_3$ and $E_4$ represent the worst-case expected privacy loss induced by Steps 3 and 4. Computing these quantities from Proposition \ref{prop:projection} and Lemma \ref{lem:gauss_mech_tdp} yields
$$
E_3 = \frac{\epsilon_1^2}{4(\ln(1/\delta_1) + \epsilon_1)} \text{ and } E_4 = \frac{\epsilon_2^2}{16\ln(1.25/\delta_2)}
$$

Since $E_3$ and $E_4$ are defined for classic neighbors, they do not contain information about $B$. To enable comparison of our algorithm across all values of $B$, we use the switch lemma to translate all $(B, \epsilon, \delta)$ to $(2, \hat{\epsilon}, \hat{\delta})$, where $s = \ceil{2B^{-1}}$, $\hat{\epsilon} = s\epsilon$, and $\hat{\delta} = \min\{1, \frac{e^{s\epsilon}-1}{e^{\epsilon}-1}\delta\}$. Hence, an upper bound on the worst-case expected privacy loss of our algorithm with parameters $B, \epsilon_1,\delta_1,\epsilon_2,\delta_2$ is given by

$$
\mathcal{U} = \frac{\hat{\epsilon}_1^2}{4(\ln(1/\hat{\delta}_1) + \epsilon_1)} + \frac{\hat{\epsilon}_2^2}{16\ln(1.25/\hat{\delta}_2)} \in [0, \infty) \cup \{\infty\}
$$
To facilitate consistent interpretation with singling-out and attribute inference protection (each of which take values in $[0,1]$), we define the distinguishing protection as $\mathcal{D} = (\mathcal{U} + 1) ^{-1} \in [0,1]$.

\section{Supplementary Figures}
\label{figs_tables}
\begin{table}[hbt!]

\caption{Summary statistics of the datasets used in our empirical analyses.}
\label{table:data}
\centering
\small
\begin{tabular}{@{}lll@{}}
\toprule
\textbf{}                                                        & \textbf{Togo}   & \textbf{Nigeria}     \\ \midrule
\textbf{Number of Subscribers}                                                   & 4,201 & 20,788 \\
\textbf{Number of Features}                                          & 10                             & 15                           \\
\textbf{Targeting Variable}                                          & Consumption                             & Loan Repayment                          \\
\textbf{Targeting Criteria}                                          & $\le 29^{th}$ percentile                            & Low-risk borrower (1)                          \\
\textbf{Range of Targeting Variable}                                           & [0.005, 1]                  & \{0,1\}               \\
\textbf{Mean of Targeting Variable}                                           & 0.073                  & 0.633               \\
\textbf{SD of Targeting Variable}                                           & 0.068                   & 0.482                        \\
\textbf{Skewness of Targeting Variable}                                           & 4.207                   & -0.551                        \\
\bottomrule
\end{tabular}

\end{table}



\end{document}